%% file: sigproc-sp.tex
\documentclass{sig-alternate-05-2015}

\usepackage{hyperref}
\usepackage[usenames]{color}
\usepackage{epsfig}
\usepackage{bbm}
\usepackage{graphicx}
\usepackage{amsmath, amssymb}
\usepackage[ruled]{algorithm2e}
\usepackage{empheq}
\usepackage{framed}
\usepackage{enumitem}
\usepackage{letltxmacro}
\usepackage{subfig}
\usepackage{multirow}
\usepackage[table]{xcolor}
\newdef{Claim}{Claim}
\newdef{Example}{Example}

\newtheorem{cor}{Corollary}
\newtheorem{prop}{Proposition}
\newtheorem{lemma}{Lemma}
\newtheorem{thm}{Theorem}

\newcommand{\prob}{\mathbb P}
\newcommand{\E}{\mathbb E}
\newcommand{\algo}{dandelion}
\newcommand{\Algo}{{\sc Dandelion}}
\DeclareMathOperator*{\argmin}{arg\,min}
\DeclareMathOperator*{\argmax}{arg\,max}

\begin{document}

\title{Dandelion: Redesigning the Bitcoin Network for Anonymity}

\numberofauthors{3} %  in this sample file, there are a *total*
\author{
\alignauthor Shaileshh Bojja Venkatakrishnan \\
\affaddr{University of Illinois at Urbana-Champaign} \\
\email{bjjvnkt2@illinois.edu}
\alignauthor Giulia Fanti \\
\affaddr{University of Illinois at Urbana-Champaign} \\
\email{fanti@illinois.edu}
%\affaddr{Depart} \\
%\affaddr{Dept. of EECS} \\
%\affaddr{University of California, Berkeley} \\
%\affaddr{Berkeley, CA} \\
%\email{fanti@berkeley.edu}
\alignauthor {Pramod Viswanath}\\
%\affaddr{Dept. of ECE} \\
\affaddr{University of Illinois at Urbana-Champaign} \\
%\affaddr{ACM} \\
%\affaddr{Champaign, IL} \\t
\email{pramodv@illinois.edu}
}

%\author{
%% You can go ahead and credit any number of authors here,
%% e.g. one 'row of three' or two rows (consisting of one row of three
%% and a second row of one, two or three).
%%
%% The command \alignauthor (no curly braces needed) should
%% precede each author name, affiliation/snail-mail address and
%% e-mail address. Additionally, tag each line of
%% affiliation/address with \affaddr, and tag the
%% e-mail address with \email.
%%
%% 1st. author
%\alignauthor
%}

% There's nothing stopping you putting the seventh, eighth, etc.
% author on the opening page (as the 'third row') but we ask,
% for aesthetic reasons that you place these 'additional authors'
% in the \additional authors block, viz.

% Just remember to make sure that the TOTAL number of authors
% is the number that will appear on the first page PLUS the
% number that will appear in the \additionalauthors section.

\maketitle
\begin{abstract}
Bitcoin and other cryptocurrencies have surged in popularity over the last decade. 
Although Bitcoin does not claim to provide anonymity for its users, it enjoys a public perception of being a `privacy-preserving' financial system.
In reality, cryptocurrencies publish users' entire transaction histories in plaintext, albeit under a pseudonym; this is required for transaction validation. 
Therefore, if a user's pseudonym can be linked to their human identity, the privacy fallout can be significant.
Recently, researchers have demonstrated deanonymization attacks that exploit weaknesses in the Bitcoin network's peer-to-peer (P2P) networking protocols.
In particular, the P2P network currently forwards content in a structured way that allows observers to deanonymize users.
In this work, we redesign the P2P network from first principles with the goal of providing strong, provable anonymity guarantees.
We propose a simple networking policy called \Algo, which achieves nearly-optimal anonymity guarantees at minimal cost to the network's utility.
We also provide a practical implementation of \Algo~for deployment.  
\end{abstract}

%A category including the fourth, optional field follows...
%\category{D.2.8}{Software Engineering}{Metrics}[complexity measures, performance measures]


%\keywords{ACM proceedings, \LaTeX, text tagging} % NOT required for Proceedings

\input{intro}

\input{model}
\input{metric}
\input{algorithm}

\input{main}

\input{systems}
\input{related}
\input{conclusion}

%\end{document}  % This is where a 'short' article might terminate

%ACKNOWLEDGMENTS are optional
%\section{Acknowledgments}

\newpage
\bibliographystyle{abbrv}

%{\scriptsize
%\bibliography{sigproc}}  % sigproc.bib is the name of the Bibliography in this 

\bibliography{privacy}

%case
% ACM needs 'a single self-contained file'!
%\balancecolumns

\appendix

\input{appendix}

\balancecolumns
% That's all folks!
\end{document}

%% file: intro.tex
\section{Introduction}
Cryptocurrencies are digital currencies that provide cryptographic verification of transactions.
Bitcoin is the best-known example of a cryptocurrency \cite{bitcoin}.
In recent years, cryptocurrencies have transitioned from an academic research topic to a multi-billion dollar industry  \cite{coinmarketcap}.

Cryptocurrencies exhibit two key properties: egalitarianism and  transparency.
%An important property of cryptocurrencies is their \emph{distributed} nature.
In this context, \emph{egalitarianism} means that no single party wields disproportionate power over the network's operation.
This diffusion of power is achieved by asking other network nodes (e.g., other Bitcoin users) to validate transactions, instead of the traditional method of using a centralized authority for this purpose.
Moreover, all transactions and communications are managed over a fully-distributed, peer-to-peer (P2P) network.
%This property has led many to hypothesize that cryptocurrencies will ``democratize  banking". %, as well as distributed computing as a whole \cite{}.

Cryptocurrencies are \emph{transparent} in the sense that all transactions are verified and recorded with cryptographic integrity guarantees; %;
this prevents fraudulent activity like double-spending of money.
Transparency is achieved through a combination of clever cryptographic protocols and the publication of transactions in a ledger known as a \emph{blockchain}.
This blockchain serves as a public record of every financial transaction in the network.

A property that Bitcoin does \emph{not} provide is anonymity. %;
Each user is  identified in the network by a public, cryptographic key.
%These keys, and every transaction associated with them, are published in the public blockchain.
%A similar case holds for most spinoff cryptocurrencies, known as altcoins.
%In practice, these pseudonyms can sometimes be linked to a human identity,
%thereby revealing the .
If one were to link such a key to its owner's human identity, the owner's entire financial history could be learned from the public blockchain.
In practice, it is possible to link public keys to identities through a number of channels, including the very networking protocols on which Bitcoin is built \cite{biryukov}.
This is a massive privacy violation, and can be downright dangerous for deanonymized users.

Bitcoin is often associated with anonymity or privacy in the public eye, despite explicit statements to the contrary in the original Bitcoin paper  \cite{bitcoin}.
People may therefore use Bitcoin without considering the potential privacy repercussions \cite{androulaki2013evaluating}.
Moreover, this problem is not unique to Bitcoin; many spinoff cryptocurrencies (known as \emph{altcoins}) use similar technologies, and therefore suffer from the same lack of anonymity in their P2P networks.

The objective of this paper is to redesign the Bitcoin networking stack from {\em first principles} to \emph{prevent network-facilitated
deanonymization} of users.
Critically, this redesign must not reduce the network's reliability or performance.
Although the networking stack is only one avenue for deanonymization attacks, it is an avenue that is powerful, poorly-understood, and often-ignored.
To better convey the problem, we begin with a brief primer on Bitcoin and its networking stack.
%We then present the problem and our contributions.

%Anonymity is often important

\subsection{Bitcoin Primer}
Bitcoin represents each user and each unit of Bitcoin currency by a public-private key pair.
A user ``possesses'' a coin by knowing its private key.
Any time a user Alice wishes to transfer her coin $m$ to Bob, she generates a signed \emph{transaction} message, which states that Alice (denoted by her public key) transmitted $m$ (denoted by its public key) to Bob (denoted by his public key).
This transaction message is broadcast to all active Bitcoin nodes, at which point \emph{miners}, or nodes who choose to help validate transactions, race to append the transaction to a global ledger known as the \emph{blockchain}.
Specifically, each miner aggregates a group of transaction messages into a \emph{block}, or list, and then completes a computational proof-of-work for the block;
the first miner to complete a proof-of-work appends their block to the blockchain and reaps a reward of newly-minted bitcoins and transaction fees.

\subsubsection{Bitcoin message propagation}
This paper focuses on one key step in the pipeline:  broadcasting transactions to other nodes.
The broadcasting process is critical because it affects which nodes can reap a transaction's mining reward (by virtue of the delivery delays to different nodes),
and it also affects the global consistency of the network (e.g., if only a subset of the users receive a given transaction).

To understand the mechanics of broadcasting, note that cryptocurrencies can be abstracted into two layers: the application layer and the network layer.
The application layer handles tasks like transaction management, blockchain processing, and mining.
Nodes are identified by their public keys in the application layer.
The network layer handles communication between nodes, which occurs over a P2P network of inter-node TCP connections.
In the network layer, nodes are identified by their IP addresses.
As we shall see momentarily, a node's IP address and public key should remain unlinkable for privacy reasons.
%The P2P network is formed by having each node establish up to eight TCP connections with other nodes. %, whose IP addresses are stored either in a global DNS server or in users' local address books.

Bitcoin's peer-to-peer broadcast of transactions and blocks is based on flooding
information along links in the P2P network.
When a node learns of a new transaction or block, it passes the message to
its neighbors who have not yet seen the message with an independent, exponential delay.
%Neighbors who do not yet have the message request it, and receive the full message.
%an \texttt{INV} message containing
%the item's hash to each of its neighbors.
%In response, if a given neighbor does not yet have that item, it requests
%it with a \texttt{GETDATA} message. The original peer responds
%with a \texttt{TX} or \texttt{BLOCK} message containing the relevant data.
%Finally, because those neighbors have learned about a
%new transaction or block, the
The process continues recursively
until all reachable peers receive the message.
This broadcast protocol is commonly known as a \emph{diffusion process};
it forms the basis
of Bitcoin's global, eventually consistent log and is
therefore of utmost importance to its correct and fair operation.

\subsubsection{Desirable Network Properties}
%\anote{what are the properties we require, and why is this so important? A) forking attacks and double spends, B) economic consequences of unreliable service, C) aligned incentives}
Bitcoin's network layer should exhibit two principal properties: low latency and anonymity.

\noindent \textbf{Low latency} means that the maximum time for a message to reach all network nodes should be bounded and small.
Latency matters because if the network fails to deliver messages within a predictable time bound, the network risks reaching an inconsistent state.

\noindent \textbf{Anonymity} means that the adversary should be unable to link transaction messages (and hence, the associated public keys) to the IP address that originated a transaction.
Recall that every transaction made by a public key is listed in plaintext in the blockchain.
Therefore, if a public key can be linked to an IP address, the adversary can link all of that user's transactions.
In some cases, the IP address could even be used to learn a node operator's human identity.
Thus, deanonymization attacks can result in a user's entire banking history being revealed.
Cryptocurrency users are typically recommended to choose fresh public keys and ``mix'' their coins with others to obscure their transaction history~\cite{coinshuffle,mixcoin} (in practice, few users do so~\cite{coinseer,fistful}). However, these techniques are useless if the IP address of the source of the transaction can be recovered.

\subsubsection{How the Current Network Fails}
In recent years, security researchers have demonstrated multiple deanonymization attacks on the Bitcoin P2P network.
These attacks typically use a ``supernode" that connects to active Bitcoin nodes and listens to the transaction traffic relayed by honest nodes \cite{koshy2014analysis,biryukov,biryukov2015bitcoin}.
Because nodes diffuse transactions symmetrically over the network, researchers were able to link Bitcoin users' public keys to their IP addresses with an accuracy of up to 30\% \cite{biryukov}.
Moreover, the source estimators used in these papers are simple, and exploit only minimal knowledge of the P2P graph structure and the structured randomness of diffusion.
We hypothesize that even higher accuracies may be possible with more sophisticated estimation tools.

These attacks demonstrate that  Bitcoin's networking stack is inadequate for protecting users' anonymity.
Moreover, the Bitcoin networking codebase is copied almost directly in other cryptocurrencies,
so the problem pervades the ecosystem.
To some extent, this is to be expected:
Bitcoin's networking stack was taken directly from Satoshi Nakamoto's original network implementation, which  appears to have been a product of expediency.
However, in the decade since Bitcoin's release,  the networking stack and its  anonymity properties have not been systematically studied.
%Our goal is to redesign it with anonymity in mind.

\subsection{Problem Statement and Contributions}
%Existing anonymity solutions for Bitcoin typically rely on novel cryptographic protocols implemented at the application layer \cite{}.
%However, cryptographic techniques are often computationally intensive \cite{},
%and they do not prevent information leakage from lower-layer networking protocols.
We aim to address the Bitcoin P2P network's poor anonymity properties through a ground-up redesign of the networking stack.
We seek a network management policy that exhibits two properties: (a) strong anonymity against an adversarial group of colluding nodes (which are a fraction $p$ of the total network size), and (b) low broadcasting latency.
We define these notions formally in Section \ref{sec:model}.
Critically, these  networking protocols should be {\em lightweight} and  provide {\em statistical anonymity guarantees against computationally-unbounded adversaries}.
Lightweight statistical solutions are complementary to cryptographic solutions, which aim to  provide worst-case anonymity guarantees, usually in the face of computationally bounded adversaries.  
Lightweight anonymization methods lower the barrier to adoption since a more efficient, faster protocol leads to a better user experience and also places less burden on developers to significantly modify existing code;  their study is also of basic scientific and engineering interest. Such is the goal of this paper.

Part of the novelty of our work is that the Bitcoin P2P networking stack has not been modeled in any detailed way (much less analyzed  theoretically), to the best of our knowledge.
In addition to modeling this complex, real-world networking system, our contributions are threefold:

\noindent \textbf{(1) Fundamental anonymity bounds.}
 The act of user deanonymization can be thought of as classifying transactions
to source nodes. Precision and recall are natural  performance metrics. % (defined precisely in Sec. \ref{sec:model}):
Recall is simply the probability of detection, a common anonymity metric that captures completeness of the estimator, whereas precision captures the exactness.
%\red{Recall is simply the probability of detection, a common anonymity metric that captures correctness, while precision captures completeness of the classifier (allowing us to quantify the scale of deanonmization being committed on the entire population).}
We define these terms precisely in Section \ref{sec:anon_metric}.
%The act of user deanonymization can be thought of as classifying transactions to source nodes.
%%To measure the adversary's ability to classify transactions, we need a metric.
%We start by highlighting the weaknesses of a common anonymity metric (probability of detection), and propose an alternative metric based on precision and recall.
%Intuitively, precision and recall are able to capture both correctness \emph{and} completeness of a classifier, whereas probability of detection only captures correctness.

Given a networking protocol, the adversary has a region of feasible (recall, precision) operating points, which are achieved by varying the source classification algorithm.
We give fundamental bounds on the best precision and recall achieved  by the adversary for any networking protocol, as illustrated in Figure \ref{fig:region};
here $p$ refers to the ratio of colluding nodes to the total number of nodes in the network.
We show that a (recall, precision) point is feasible only if it lies between the red and blue lines in Figure \ref{fig:region}.
Moreover, every networking protocol yields an achievable (recall, precision) region to the adversary that intersects with the shaded region (a) in Figure \ref{fig:region} in at least one point. 
%(here $p$ refers to the ratio of colluding nodes to the total number of nodes in the network). 

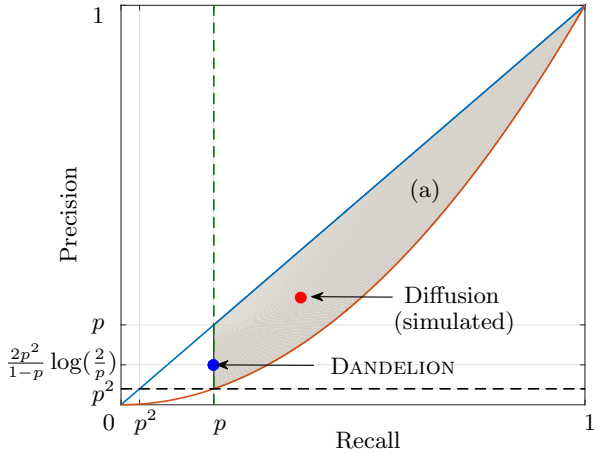
\begin{figure}[h]
    \centering
  \includegraphics[width=.41\textwidth]{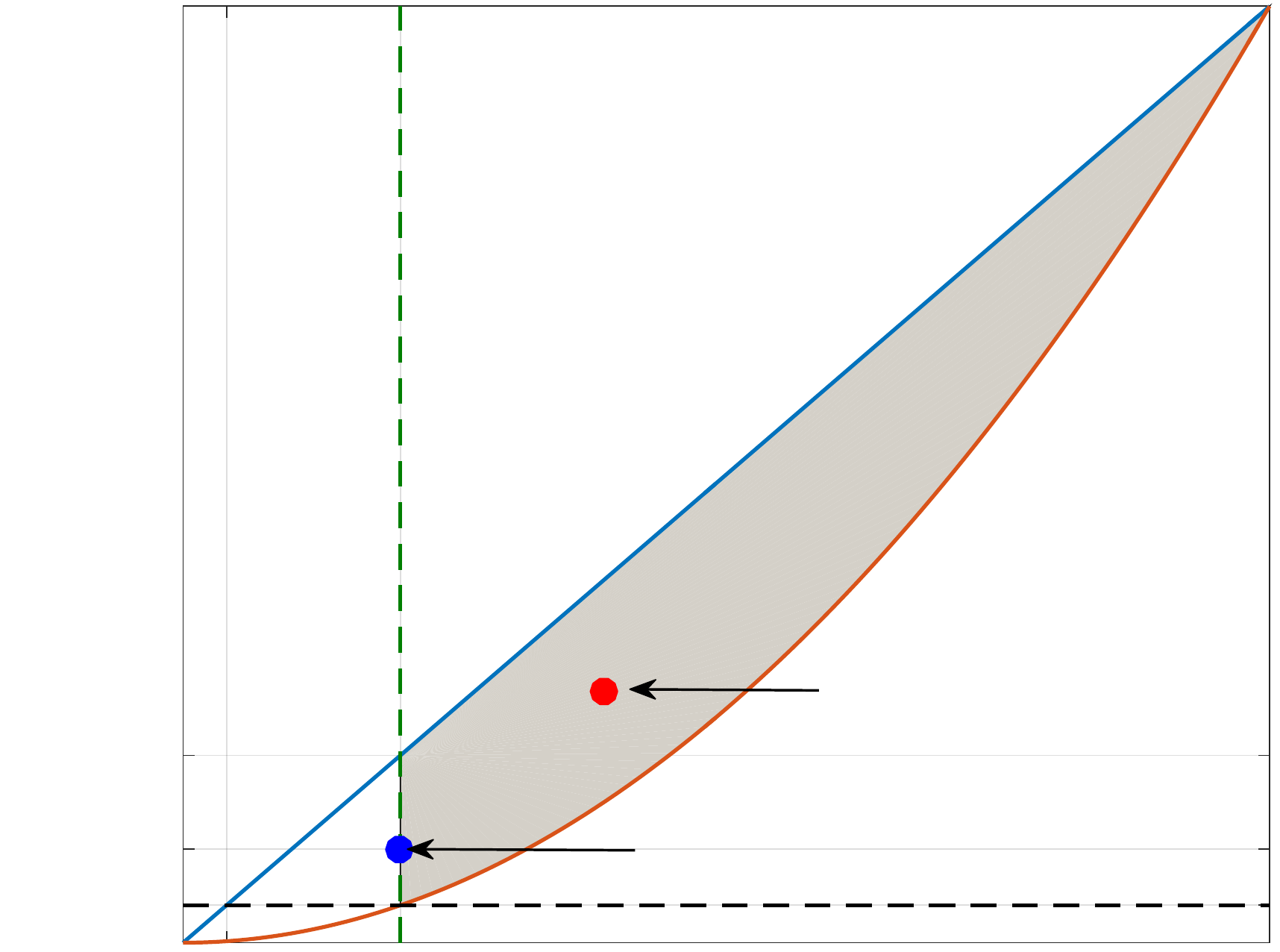}
  \put(-68,80){(a)}
  \put(-96,-15){Recall}
  \put(-184,-8){0}
  \put(-172,-8){$p^2$}
  \put(-142,-8){$p$}
  \put(-2,-8){$1$}
  % vertical
  \put(-188,3){$p^2$}
  \put(-98,13){\Algo}
  \put(-70,40){Diffusion}
  \put(-74,30){(simulated)}
  \put(-221,15){\rotatebox{0}{$\frac{2p^2}{1-p}\log(\frac{2}{p})$}}
  \put(-188,30){$p$}
  \put(-188,146){$1$}
  \put(-200,55){\rotatebox{90}{Precision}}
  \caption{Bounds on the precision and recall of any networking protocol, plotted for $p=0.2$.
% Every networking protocol's precision-recall region intersects the shaded region (a) in at least one point.
 \Algo~ has strong anonymity properties, achieving a precision-recall region close to the fundamental lower bounds.}
  \label{fig:region}
\end{figure}

\noindent \textbf{(2) Optimal algorithm.} We propose a simple networking protocol called \Algo, whose achievable precision-recall region is nearly optimal, in the sense that it is contained in the achievable region of (nearly) every other possible networking protocol.

\Algo~ consists of two phases.
In the first phase, each transaction is propagated on a random line; that is, each relay passes the message to exactly one (random) node for a random number of hops.
In the second phase, the message is broadcast as fast as possible using diffusion.
\Algo~ has  two key constraints: (a) in the first phase, all transactions from all sources should propagate over the \emph{same} line, and (b) the adversary should not be able to learn the structure of the line beyond the adversarial nodes' immediate neighbors.

The point labeled `\Algo' in Figure \ref{fig:region} is the Pareto frontier of \Algo's achievable precision-recall region.
We compare this to the achievable region for diffusion.
The point labeled `Diffusion' was obtained by simulating a diffusion process on a snapshot of the Bitcoin server network from 2015 \cite{coinscope}, and using a suboptimal source classifier.
Because of this, the achievable region must contain the plotted point, but may be larger.
Not only is the region for diffusion a superset of the one for \Algo, but \Algo's region is nearly as small as possible.
We revisit Figure \ref{fig:region} in greater detail in Sections \ref{sec:bounds} and \ref{sec:main}.

\noindent \textbf{(3) Practical implementation.} We outline the practical challenges associated with implementing \Algo.
In particular, constructing the graph for \Algo~in a distributed fashion, and enforcing the assumption that the adversary cannot learn the graph, are nontrivial.
%For this reason, we propose a random regular graph structure that is more difficult to learn and changes over time.
%We prove that this alternative structure has similar precision and recall properties to \Algo, and we show simple methods for generating random regular graphs on the fly.
We therefore propose simple heuristics for addressing these challenges in practical implementations.

\vspace{0.1in}
\noindent \textbf{Paper Structure.}
We begin by discussing Bitcoin's P2P networking stack and our problem of interest, which we model in Section \ref{sec:model}.
We then present fundamental bounds on our anonymity metric in Section \ref{sec:bounds};
these bounds are used for comparison with various networking policies later in the paper.
In Section \ref{sec:algos}, we present some first-order solutions, and explain why they do not work.
We present our main result, \Algo, in Section \ref{sec:main}.
Section \ref{sec:systems} discusses the systems challenges of implementing \Algo, and proposes some simple, heuristic solutions.
We discuss the relation between \Algo and prior related work in Section \ref{sec:related}, and conclude with some open problems in Section \ref{sec:conclusion}.

%% file: model.tex
\section{System Model} 
\label{sec:model}
We model three critical aspects of Bitcoin's P2P network: 
the network topology, the message propagation protocol, and the deanonymizing adversary's capabilities.
These models are based on existing protocols and observed behavior.

\subsection{P2P Network Model}
The Bitcoin P2P network contains two classes of nodes: servers and clients. 
Clients are nodes that do not accept incoming TCP connections (e.g., nodes behind NAT), whereas
servers do accept incoming connections.
%Clients and servers have different networking protocols and anonymity concerns. 
%For instance, clients do not relay transactions.
We focus in this work on servers because (a) they are more permanent in the network, and (b) it is straightforward to generalize server-oriented anonymity solutions to also protect clients.

We model the P2P network as a graph $G(V,E)$, where $V$ is the set of all server nodes and $E$ is the set of edges, or connections, between them.
%For a fixed topology $T$, we assume that the nodes are equally likely to assume each possible label ordering in $T$;
% $G(V,E)$ describes the resulting, labeled graph.
%; we denote the resulting graph using the random variable $G$. 
For a node $v$, $\Gamma(v)$ denotes the set of $v$'s neighbors in $G$. 
Similarly for a set of nodes $U$, $\Gamma(U)$ denotes the set of all neighborhood sets of the nodes in $U$. 
To model the graph's topology, we first discuss Bitcoin's network management protocols.

Each node in the Bitcoin P2P network has an \emph{address manager}---a list of other nodes' contact information represented as a (IP address, port) pair, along with a time estimate of when that node was last active.
When a server first joins the network, its address manager is empty, but 
the node can learn a random set of active addresses by contacting a hard-coded DNS server.
During normal network operation, nodes periodically relay entries from their address managers, 
which helps spread information regarding active peers.
We model address managers by assuming that each node possesses the contact information for \emph{every} other Bitcoin server.
In practice, address managers actually contain a random sample of population IP addresses.

Each server is allowed to establish up to eight outgoing connections 
to nodes in the server's address manager. 
An \emph{outgoing connection} from Alice to Bob is one that is initiated by Alice. 
However, these TCP connections are bidirectional once established. 
%Each server is also allowed to accept up to 125 incoming connections, 
%but the majority of incoming connections originate from clients. 
%Since we are focused on servers, we only consider connections in which both endpoints are servers. 
We therefore model the subgraph of server-to-server connections as a random 16-regular graph. 
In practice, the degree distribution is not quite uniform---we revisit this issue in Section \ref{sec:systems}.

\subsection{Transaction Model}
\label{sec:trans_model}
%Any node can complete a transaction at any time. 
As explained in Section \ref{sec:adv_model}, the network is partitioned into honest nodes
and colluding, adversarial nodes, who attempt to deanonymize users.
In this work, we assume that all honest nodes generate one transaction in the time period of interest.
In practice, servers generate transactions at different rates;
however, all transactions by a single node are identified by the node's public key (as long as the node does not generate fresh keys).
Therefore, we treat multiple transactions from the same node as a single transaction to be deanonymized.
We also assume the exact time when each server starts broadcasting its transaction is unknown to the adversary. %\red{should we do a short discussion on the timescale of observation in practice?}
A typical transaction can take up to 60 seconds to propagate through the Bitcoin network \cite{decker2013information}, 
so estimating its time of origin at a useful granularity of a second or sub-second can be difficult. 
$\mathcal{X}$ is the set of all transaction messages from honest servers. 
$X_v$ is the transaction message originating from honest server $v$ and
$\mathbf{X}$ is a vector containing the ground truth mapping between source nodes $v$ and transactions $X_v$. 
We model the mapping between servers and transaction messages as being drawn uniformly from the set of all such mappings.

\noindent \textbf{Spreading Model.}
Once a Bitcoin transaction is complete, the source broadcasts the transaction message over the network. 
The protocol for broadcasting transactions should ensure low transaction latency, in order to provide network consistency and fairness.
%We let $T_v(X_w)$ denote the time at which node $v$ receives transaction $X_w$. 
%The propagation mechanism for transactions should ensure that the \textit{average maximum latency} $\frac{1}{|V|}\sum_w[\max_{v} T_v(X_w)]$ is small.
%\red{is this what we want to say? Or maybe max max?}

Bitcoin currently uses a \emph{diffusion} propagation mechanism to broadcast transactions,
in which each transaction source or relay passes the transaction to the node's neighbors with independent, exponential delays.
Once a node has received a particular transaction, the node does not accept future relays of the transaction.
This diffusion spreading serves as a baseline for our algorithmic improvements.
It has good latency properties due to its exponential spreading \cite{bartlett1956deterministic}. %, \red{and has an expected per-node maximum latency of XXX on XXX class of graphs.}

More generally, in this work, we consider spreading policies that are symmetric in the neighbor node IDs; 
that is, a forwarding node does not use the IP address values (or other metadata) of its neighbors to influence its forwarding decisions. 
This holds for diffusion spreading, but we constrain our proposed solutions to also satisfy the same property.

\subsection{Adversarial Model} 
\label{sec:adv_model}
We consider an adversary whose goal is to deanonymize users by linking their transactions (and hence, their public keys) to their IP addresses.
In particular, we are interested in defending against \emph{botnets}---large sets of malware-infected hosts that are controlled remotely, often without the host owners' knowledge \cite{silva2013botnets}.
Botnets are a commonly-studied adversarial model for various Bitcoin attacks \cite{apostolaki2016hijacking}, 
largely because they are easy to access, cheap, and pervasive in the Bitcoin network \cite{plohmann2012case}.
While botnets  can have many uses, we wish to defend against a botnet that aims to deanonymize users.

We model the botnet adversary as a set of adversarial, colluding ``spy" nodes that participate in the Bitcoin network as if they were honest nodes (i.e., honest-but-curious). 
We denote honest nodes by $V_H$ and adversarial nodes by $V_A$.
For a parameter $p$, we assume a fixed number of adversarial nodes ($|V_A|=np$) and honest nodes ($\tilde{n} = |V_H| = (1-p)n$). 
The adversarial nodes are dispersed uniformly at random in the network; 
this reflects the botnet's ability to obtain IP addresses uniformly across the IP address space.  
However, for a given topology, the actual locations of the honest/adversarial nodes are random. 
%\red{We further assume the adversary knows the IP addresses of all nodes in $V$, 
%whereas honest nodes know only the IP addresses of their neighbors} and
We further assume that all nodes know the complete list of active IP addresses,
and honest nodes cannot distinguish between an adversarial and honest IP address.  

Whenever a transaction is broadcast over the network, the adversarial nodes log the timestamps and the honest neighbors from which they receive the transaction.
We assume a continuous-time system, in which simultaneous transmissions do not occur. %each transaction reaches one spy node strictly before the other spies; we track this timestamp for each message.
For each honest server node $v$, we let $S_v$ denote the set of (transaction, receiving spy node, timestamp) tuples $(x,u,T_u(x))$ such that transaction $x$ was forwarded by honest node $v$ to adversary $u\in V_A$ at time $T_u(x)$ (Fig. \ref{fig:spies}); 
$\mathbf{S}$ is the vector of all $S_v$'s.  
We shall see in Section \ref{sec:main} that the honest server who \emph{first} delivers a given transaction to the adversary plays a special role.
%With full generality, we assume each transaction is first delivered to the adversary by a single honest node (i.e., no simultaneous transmissions).
% We denote by $S_v$ the set of messages first relayed to the adversary  by honest server node $v$; 

\begin{figure}[t]
    \raggedleft
  \includegraphics[width=.33\textwidth]{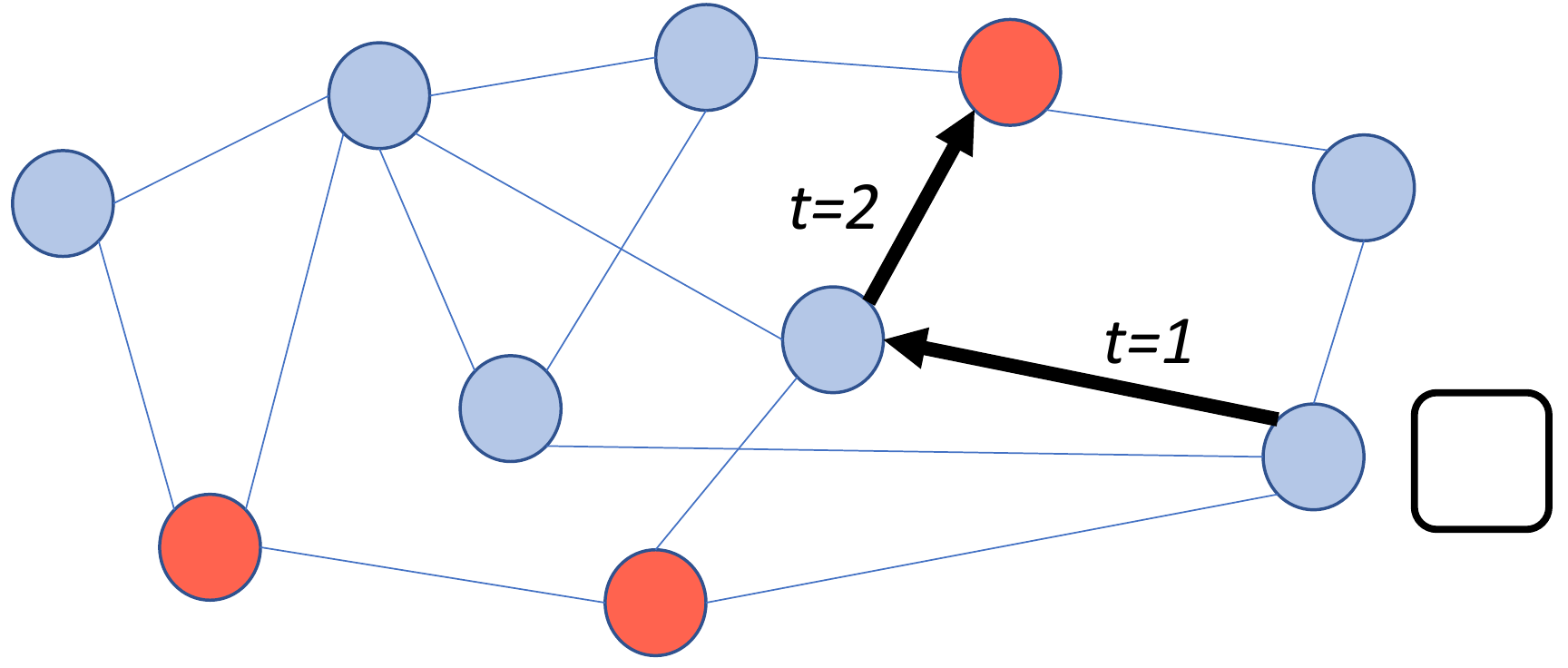}
  \put(-240,50){$|V_H| = 7$}
  \put(-240,35){$|V_A| = 3$}
  \put(-240,20){$S_u = \{(X_v, w, 2), \ldots \}$}
  \put(-15,18){$X_v$}
  \put(-30,20){$v$}
  \put(-80,32){$u$}
  \put(-63,60){$w$}
  \caption{Red nodes are adversarial spies; blue nodes are honest. Message $X_v$ reaches the spy $w$ at time $t=2$.}
  \label{fig:spies}
\end{figure}
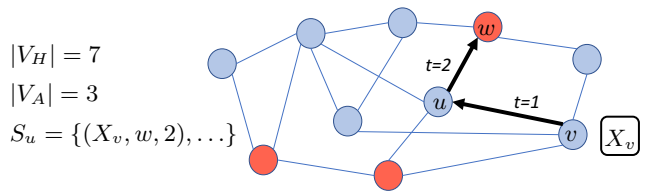

In addition to the transaction timestamps, the adversaries can also learn the network structure $G$, partially or completely, over time. 
The extent of such knowledge depends on the dynamism of the network, and will be made clear in the context of the specific networks being considered. 
For example, if the network is static over an extended period of time then adversaries can learn the entire graph $G$. 
On the other hand, in a fast changing network, the adversaries have knowledge of only their local neighborhood $\Gamma(V_A)$. 
For ease of exposition, let us, for now, use  $\mathbf{\Gamma}$ to denote the adversary's knowledge of the graph. 

Once the timestamps have been collected, the adversarial nodes collude to infer the transaction source.
The adversary uses its observations $\mathbf{O} = (\mathbf{S},\mathbf{\Gamma})$  to output a mapping between transactions and honest servers;
we let $\mathtt{M}(X_v)\in V_H$ denote the server associated with transaction $X_v$ in the adversary's mapping.
This mapping is chosen to maximize the adversary's deanonymization payoff,  defined in Section \ref{sec:anon_metric}.

\subsection{Anonymity Metric}
\label{sec:anon_metric}

A common metric for measuring a broadcasting scheme's anonymity is \textbf{probability of detection}. 
For a fixed transaction and estimator, probability of detection is defined as 
\begin{equation}
\prob_{\mathtt{M},G}(\text{detection}) = \frac{\sum_{v\in V_H} \prob(\mathtt{M}(X_v) = v)}{\tilde n},
\end{equation}
or the probability that the estimator outputs the correct source of a single transaction,
computed over all  transaction sources $v\in V_H$, mappings between sources and transactions $\mathbf X$, realizations of the message propagation trajectory, and graph realizations $G$ (if the graph is random).
While probability of detection considers a single source, our problem considers the joint deanonymization of transactions from distinct sources.
In this case, probability of detection inherently captures the {recall}, or completeness, of an estimator. %fails to capture the notions of precision and recall.
We propose to augment this metric by also studying precision, which captures the exactness of an estimator.

Precision and recall are performance metrics commonly used in information retrieval for binary classification. 
Suppose we have $n$ data items, each associated with a class: 0 or 1. 
We are given a classifier that labels each data item as either a 0 or a 1, without access to the ground truth.
We designate one of these classes (e.g. class 1)  `positive'.
For a given classifier output on a single item, a \emph{true positive} means the item was correctly assigned to class 1, 
and a \emph{true negative} means the item was correctly assigned to class 0.
A \emph{false positive} means a 0 item was incorrectly classified as a 1, and a \emph{false negative} means a 1 item was incorrectly classified as a 0.
If we run this classifier on all $n$ data items, precision and recall are defined as follows:
\begin{eqnarray*}
\textbf{Precision} &=& \frac{|\text{True Positives}|}{|\text{True Positives}| + |\text{False Positives}|} \\
\textbf{Recall} &=& \frac{|\text{True Positives}|}{|\text{True Positives}| + |\text{False Negatives}|}
\end{eqnarray*}
where $|\cdot|$ denotes the cardinality of a set, and `True Positives' denotes the set of all data items whose classification output was a true positive (and so forth).

Precision can be interpreted as the probability that a randomly-selected item with label 1 is correct, 
whereas recall can be interpreted as the probability that a randomly-selected data item from class 1 is correctly classified.
Adapting this terminology to our problem, we have a multiclass classification problem; each server is a class, and each transaction is to be classified.
For a given server $v$ and mapping $\mathtt{M}$, the precision $D_{\texttt{M}}(v)$ comparing class $v$ to all other classes is computed as\footnote{Following convention we define $D_\mathtt{M}(v)=0$ if both the numerator and denominator are $0$ in Equation~\eqref{eq:precision}.}
\begin{equation}
D_{\texttt{M}}(v) = \frac{\mathbbm{1}\{\mathtt{M}(X_v) = v\}}{\sum_{w\in V_H} \mathbbm{1}\{ \mathtt{M}(X_w) = v\} },
\label{eq:precision}
\end{equation}
and the recall is computed as
\begin{equation}
R_{\texttt{M}}(v) = \mathbbm{1}\{\mathtt{M}(X_v) = v\}
\label{eq:recall}
\end{equation}
where $\mathbbm{1}\{\cdot\}$ denotes the indicator function.
In multiclass classification settings, precision and recall are often aggregated through \emph{macro-averaging}, which consists of averaging precision/recall across classes.
This approach is typically used when the number of items in each class is equal \cite{sebastiani2002machine}, as in our problem.
We therefore average the precision and recall over all servers and take expectation, giving an expected macro-averaged precision of $\E[D_\mathtt{M}] = \frac{1}{\tilde n}\sum_{v\in V_H} \E[D_{\texttt{M}}(v)]$ and recall of $\E[R_\mathtt{M}] = \frac{1}{\tilde n}\sum_{v\in V_H} \mathbb{E}[R_{\texttt{M}}(v)]$.
% $P$ and $R$ not defined previously.

We now explain why probability of detection does not capture the distinction between precision and recall. 
First, note that  the expected per-node recall is identical to the probability of detection.
Now consider two estimators: in the first, the adversary's strategy is to assign all $\tilde n$ transactions to one randomly-selected server $v$.
In the second, the adversary creates a random matching between the $\tilde n$ transactions and honest servers.
Both estimators have a probability of detection (i.e., expected per-node recall) of $1/\tilde n$.
However, the first estimator has an expected per-node precision of $1/\tilde n^2$, while the second has an expected per-node precision of  $1/\tilde n$.
Operationally, this can be interpreted as a difference in plausible deniability: 
the implicated node $v$ in the first case can deny being the source of any given transaction, because it could not have generated all $\tilde n$ transactions.
If a node is  correctly implicated in the second estimator, it has no plausible deniability.
Probability of detection alone does not capture this difference, and is therefore insufficient as a standalone metric.

In this work, we quantify anonymity through a combination of \textbf{expected macro-averaged precision} (or ``precision" for short) and \textbf{expected macro-averaged recall} (or ``recall", or probability of detection).
Higher precision and recall favor the adversary. 
%As we explain in Section \ref{}, we will consider a class of solutions that naturally achieves an optimal probability of detection (i.e., expected per-node recall), so we do not need to jointly optimize over both metrics.
%\noindent \textbf{K-anonymity} originally came from the study of sanitized databases.
%It states the 
For a mapping strategy \texttt{M} let $D_\mathtt{M}$  and $R_\mathtt{M}$ denote the average precision and recall, respectively, obtained in a realization. 
Our metrics of interest, then, are the overall expected precision $\mathbf{D}_\mathtt{M} = \mathbb{E}[D_\mathtt{M}]$ and recall $\mathbf{R}_\mathtt{M} = \mathbb{E}[R_\mathtt{M}]$. 
This expectation is taken over four random variables: the graph realization $G$ (which can be random in general), the mapping between servers and messages $\mathbf X$, the observed timestamp and topological information $\mathbf{O}$, and the adversary's mapping strategy $\mathtt{M}$.
Similarly let $D_\mathtt{M}(v)$ and $\mathbf{D}_\mathtt{M}(v)$ denote the instantaneous and expected precisions at a server $v\in V_H$,
and let $R_\mathtt{M}(v)$ and $\mathbf{R}_\mathtt{M}(v)$ denote the instantaneous and expected recalls.
Let $\mathbf{D}_\mathtt{OPT}$ and $\mathbf{R}_\mathtt{OPT}$ denote the precision and recall, respectively, of the precision-maximizing and recall-maximizing mapping strategies, respectively.
The optimal precision is not necessarily achieved by the same mapping strategy as the optimal recall.
The adversary is computationally unbounded.

\subsection{Problem Statement}
As network designers, we control two aspects of the network: the graph creation/maintenance strategy and the spreading protocol.
Our goal is to choose a graph-selection strategy and a spreading protocol that simultaneously give low average latency, precision, and recall guarantees.
We restrict ourselves to the following model of graph generation:
For a fixed topology $\tau$, we assume that the nodes are equally likely to assume each possible label ordering in $\tau$.
Moving forward, $G(V,E)$ will describe the resulting, labeled graph. %, which is used only in the anonymity phase.

%To simplify the solution space (both in terms of analysis and implementation), we consider a class of algorithms that operate in two phases: an anonymity phase aimed at hiding the source, and a spreading phase aimed at fast propagation.
%We assume in this work that the spreading phase propagates content over the existing Bitcoin P2P network using the same diffusion protocol that is currently implemented,
%which is known to have good latency properties.
%We focus instead on designing an anonymity phase that minimizes the adversary's precision and recall, while adding only a constant-order delay to each message's overall latency.
%This anonymity phase is allowed to use an entirely distinct graph from the main Bitcoin P2P network (albeit still defined over the same $n$ nodes), as well as a separate spreading protocol.

Let $\mathcal T$ denote the set of all graph topologies over $n$ nodes, and $\Sigma$ the set of graph-independent spreading strategies.
The adversary controls only the estimation algorithm for mapping transactions to nodes. 
Given a topology $\tau \in \mathcal T$ and a spreading strategy $\sigma \in \Sigma$, let $\mathcal {M}_{\tau,\sigma}$ denote the set of mapping strategies that map $\tilde n$ transactions to $\tilde n$ servers, with all knowledge derived from the topology and the spreading strategy. If $\tau$ and $\sigma$ are clear from context, we simply use $\mathcal{M}$ to denote the space of mapping strategies. 
We define the \emph{detection region} for $\tau$  and $\sigma$ as the set of achievable
precision and recall operating points:
\[
\Omega(\tau, \sigma) = \{(D,R) ~ | ~ \exists ~ \texttt{M}\in \mathcal M_{\tau, \sigma},  ~D = \mathbf{D}_\mathtt{M},  R = \mathbf{R}_\mathtt{M}\}.
\]
Note that the detection region always contains the origin. 
The adversary's goal is to find estimators that achieve the boundary points of the region, whereas our goal is to make the detection region as small as possible.

\noindent \textbf{Problem:} Characterize fundamental, protocol-independent bounds on the detection region. Further, identify a $(\tau^*,\sigma^*)$ pair whose detection region is a subset of the detection region of every  graph-generation and spreading strategy:
\begin{equation}
\Omega(\tau^*, \sigma^*) = \bigcap_{\sigma \in \Sigma, \tau \in \mathcal T} \Omega(\tau, \sigma).
\label{eq:intersection}
\end{equation}
%Intuitively, this is trying to find a networking strategy that minimizes the adversary's achievable range of precision and recall operating points.
It is unclear a priori if such a strategy pair exists.
In this work, we show a simple networking policy that closely approximates condition \eqref{eq:intersection}.

%% file: metric.tex
\section{Anonymity Metric Properties}
\label{sec:bounds}
Precision and recall are not generally used as anonymity metrics, since most anonymity systems provide per-user ano-nymity guarantees \cite{tor,tarzan,reiter1998crowds,KFSV14}.
We instead want guarantees against a stronger adversary that jointly deanonymizes multiple users.
The goal of this section is to give intuition about precision and recall as metrics, and to provide fundamental bounds on both. 

Our problem differs from traditional classification in that there is only one data item (transaction) per class (server). 
This restricts the set of achievable macro-averaged precision-recall points in a somewhat unconventional way. 
We first explain how precision and recall are typically used, and then prove fundamental bounds that illustrate the ways in which our problem differs from traditional classification problems. % of these metrics. % on macro-averaged precision and recall.

\vspace{0.1in}
\noindent \textbf{Precision-Recall Curves.}
Most binary classifiers have an internal parameter (e.g., a threshold) that can be varied to give the classifier different precision and recall characteristics.
Sweeping this parameter typically yields a tradeoff between precision and recall. 
While this tradeoff has been studied theoretically \cite{powers2011evaluation}, it is most often illustrated empirically for a given classifier, through curves like Figure \ref{fig:region_bounds} (right).
Notably, a classifier can achieve high recall ($\approx 1$) at the expense of precision by assigning all data elements to the positive category, 
or high precision ($\approx 1$) at the expense of recall by classifying only data elements that are clearly true positives. 
Hence the precision-recall points $(0,1)$ and $(1,0)$ are typically achievable in practice.

Unlike traditional precision-recall curves, we are not interested in the curve for a single estimator; 
we want to identify the achievable detection region across \emph{all} estimators.
Moreover, since ours is a multi-class classification problem, we consider \emph{macro-averaged} precision and recall. 
With macro-averaging, increasing the recall (resp. precision) for one class will often reduce the recall (resp. precision) for another. 
Therefore, it is unclear what the precision-recall tradeoff will look like, or even if the boundary points $(0,1)$ and $(1,0)$ are achievable.
The following theorem restricts the set of feasible, macro-averaged precision-recall points for \emph{any} estimator the adversary employs.

%\noindent {\bf Precision-Recall:} 

\begin{thm} \label{thm:bounds_prec_rec}
Any mapping policy $\mathtt{M}\in\mathcal{M}_{\tau,\sigma}$ on a network with topology $\tau\in\mathcal{T}$ and spreading strategy $\sigma \in\Sigma$ has a precision and recall that are bounded as 
\begin{align}
\mathbf{D}_\mathtt{M} \overset{(a)}{\leq} \mathbf{R}_\mathtt{M} \overset{(b)}{\leq} \sqrt{\mathbf{D}_\mathtt{M}}.
\end{align}
\end{thm}
(Proof in Section \ref{proof:bounds_prec_rec})

This theorem follows from the definition of macro-averaged precision and recall;  it implies that not only are corner points $(0,1)$ and $(1,0)$ unachievable, but every estimator's detection region must lie between the blue and red lines in Fig. \ref{fig:region_bounds} (left).
Given this constraint, a natural question is whether there exist precision and recall points that can always be achieved, regardless of the networking protocol.  
We demonstrate the existence of such points by analyzing a simple estimator.

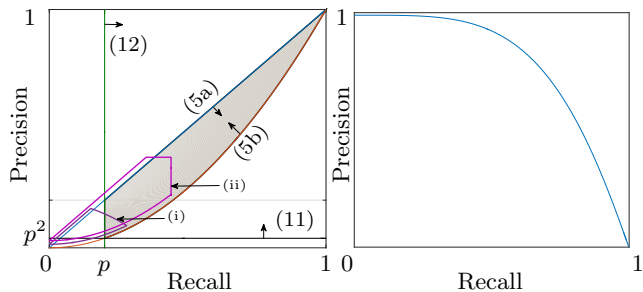
\begin{figure}[t]
    \centering
  \includegraphics[width=.21\textwidth]{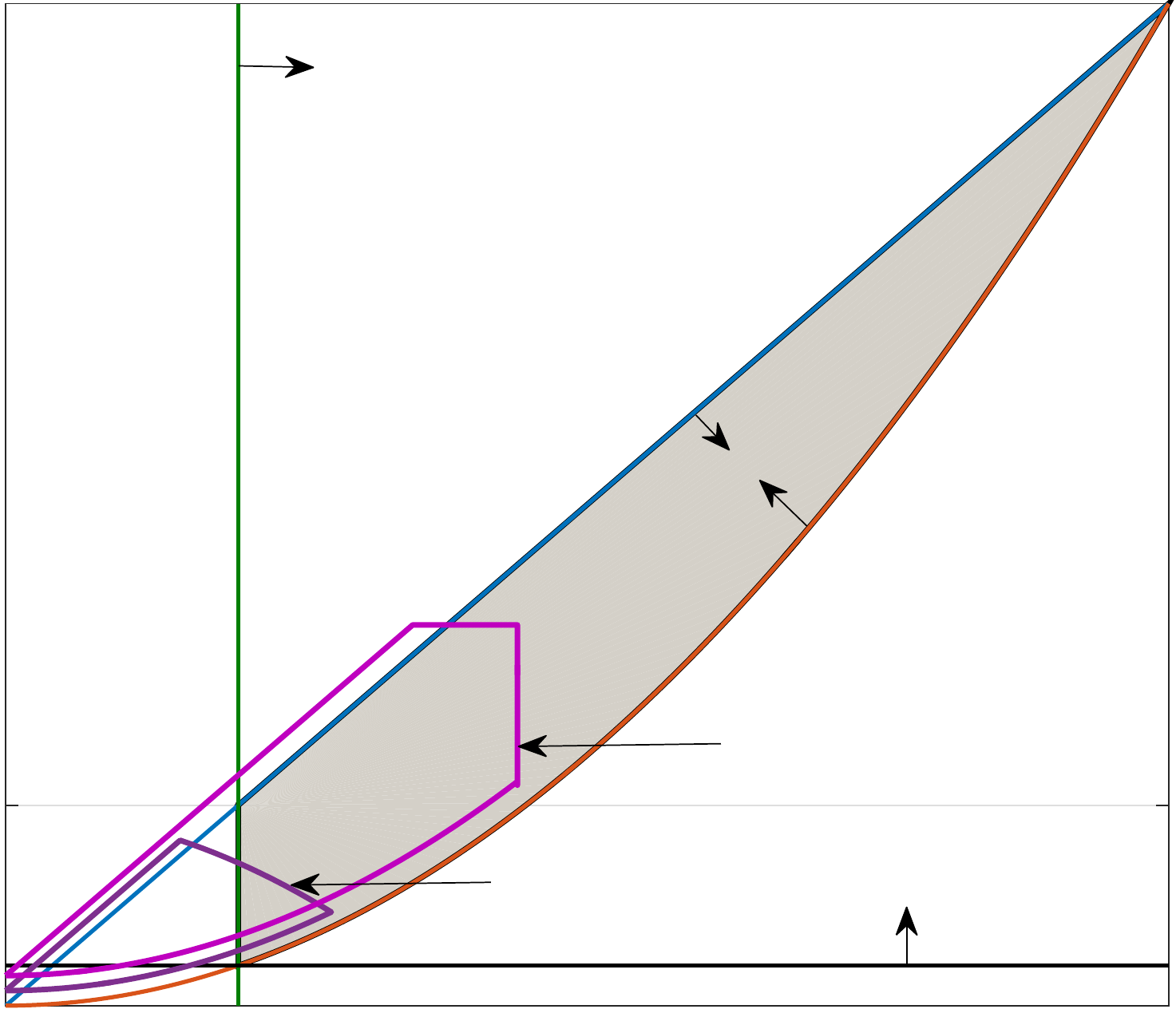}
  ~
  \includegraphics[width=.21\textwidth]{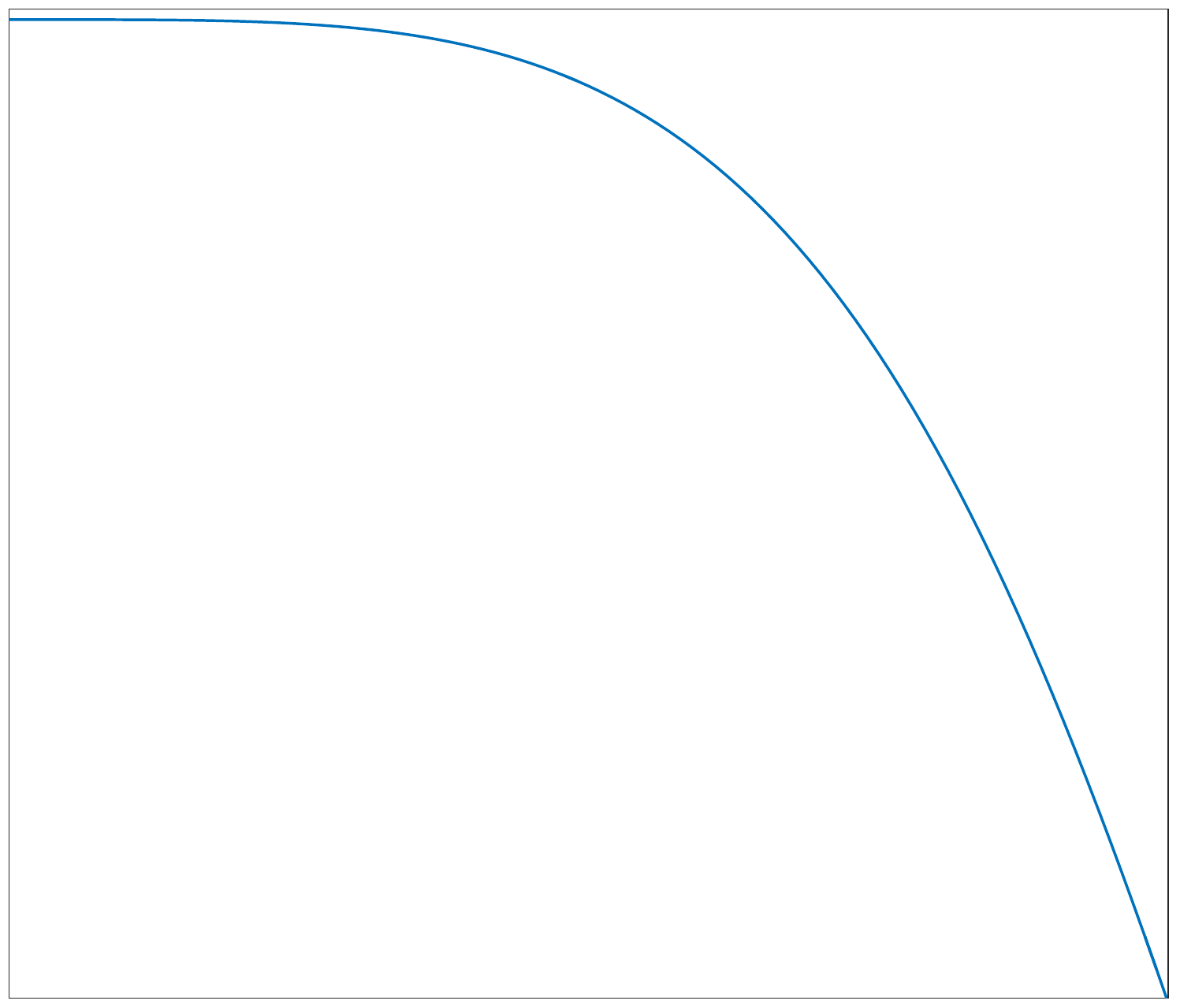}
  \put(-171,51){\rotatebox{45}{\small{(5a)}}}
  \put(-152,33){\rotatebox{49}{\small{(5b)}}}
  \put(-176,-15){Recall}
  \put(-66,-15){Recall}
  \put(-135,9){\small{(11)}}
  \put(-176,11){\tiny{(i)}}
  \put(-155,23){\tiny{(ii)}}
  \put(-198,75){\small{(12)}}
  \put(-224,-8){0}
  \put(-110,-8){0}
  \put(-202,-8){$p$}
  \put(-118,-8){$1$}
  \put(0,-8){$1$}
  % vertical
  \put(-230,3){$p^2$}
  \put(-230,85){$1$}
  \put(-112,85){$1$}
  \put(-235,25){\rotatebox{90}{Precision}}
  \put(-113,25){\rotatebox{90}{Precision}}
%  \put(-169,35){(i)}
  \caption{Bounds on the precision-recall detection region for any networking policy (left). 
 Each bound is labeled with the corresponding equation number from Section \ref{sec:bounds}. 
 Example of a typical precision-recall curve (right).
 }
  \label{fig:region_bounds}
\end{figure}

\vspace{0.1in}
\noindent \textbf{Lower Bounds.} 
Computing lower bounds on precision and recall is challenging because the adversary's knowledge can vary depending on the networking policy.
However, the so-called  \emph{first-spy estimator} (which is used in practical attacks like \cite{biryukov})  relies only on the adversary's knowledge of its local network neighborhood.
The adversaries we consider will always have access to this information.
%Our fundamental lower bounds on the maximum precision and recall of any detection region come from analyzing the .
The first-spy estimator outputs the first honest node to send a given message to any of the adversarial nodes. 
%This estimator requires no knowledge of the underlying graph. 
We start by showing that the first-spy estimator always achieves a precision and recall of at least $p^2$ and $p$, respectively, where $p$ is the fraction of spies. 
This in turn implies that the maximum precision and recall over all estimators are individually lower-bounded by $p^2$ and $p$, respectively.

\begin{thm} \label{thm:lower_bounds_fs}
The optimal precision and recall on a network with a fraction $p$ of adversaries and any spreading policy are lower bounded as
\begin{eqnarray}
\mathbf{D}_\mathtt{OPT} &\geq& p^2 \label{eq:prec_lower} \\
\mathbf{R}_\mathtt{OPT} &\geq& p  \label{eq:rec_lower}.
\end{eqnarray}
\end{thm}
(Proof in Section \ref{proof:lower_bounds_fs})

This theorem implies that for any networking policy, the detection region must include at least one point in the shaded region of Figure \ref{fig:region_bounds} (left).
Note that if an estimator can achieve a given (recall, precision) point, then it can also achieve points with elementwise lower precision and recall by choosing to discard observed information.  
The purple curves labeled (i) and (ii) outline the boundaries of two  examples of feasible detection regions, staggered for visibility. 

\vspace{0.1in}
\noindent {\bf Optimizing Estimators.}
Given these constraints on the detection region, we want to understand what estimators achieve the maximum precision and recall, respectively.
For a given network specification, precision and recall might be maximized by different estimators; 
if this is the case, then the detection region will have a nontrivial Pareto frontier, like curve (i) in Figure \ref{fig:region_bounds} (left). 
On the other hand, if the same estimator maximizes precision and recall, the detection region's Pareto frontier will be a single point, like curve (ii).

%This section assumes that the adversary only has access to $\mathbf S$, not to the exact timestamps for each spy.
%We will revisit this assumption more carefully in Section \ref{sec:algos}, but the reason is that we eventually restrict ourselves to a class of spreading protocols and graphs for which exact timestamps do not help the adversary, aside from establishing which honest node  first delivered the message to the adversary, and by way of which spy.
%As such, we will optimize over the class of estimators that use only $\mathbf S$ and $\Gamma(V_A)$ as inputs.

We start by proving that in order to maximize precision, the adversary should use a maximum-weight matching estimator, where the weights depend on the  information observed by the adversary, such as graph structure and timestamps.
\begin{thm}[Precision-Optimal Estimator] \label{thm: optimal estimator}
The \\ precision-optimizing estimator for an adversary with observations $\mathbf{O}=(\mathbf{S},\mathbf{\Gamma})$, is achieved by a matching over the bipartite graph $(V_H,\mathcal{X})$. Moreover, such a matching is a maximum-weight matching for edge weights $\mathbb{P}(X_v = x|\mathbf{O})$ on each edge $(v,x)\in V_H\times \mathcal{X}$ of the graph. 
\end{thm}
(Proof in Section \ref{proof: optimal estimator})

Theorem \ref{thm: optimal estimator} gives a corollary used in Section \ref{sec:algos} for bounding the performance of various networking protocols.
\begin{cor} \label{cor: opt upper bound}
The optimal expected payoff at a server $v$, under observations $\mathbf{O}=(\mathbf{S},\mathbf{\Gamma})$ for the adversaries, is upper bounded as
\begin{align}
\mathbb{E}[D_\mathtt{OPT}(v)|\mathbf{O}] \leq \max_{x\in\mathcal{X}} \mathbb{P}(X_v = x | \mathbf{O}). \label{eq: opt upper bound}
\end{align}
\end{cor}
(Proof in Section \ref{proof: opt upper bound})
%\begin{proof}
%Let $B\in\mathcal{B}$ be any mapping under observations $\mathbf{S},\Gamma(V_A)$. Consider a server $v$ and let $\{x_1,x_2,\ldots,x_k\}$ be the set of messages that are mapped to $v$ in $B$. Then, 
%\begin{align}
%\mathbb{E}[D_\mathtt{B}(v)|\mathbf{S},\Gamma(V_A)] &\leq \frac{\sum_{i=1}^k \mathbb{P}(X_v=x_i|\mathbf{S},\Gamma(V_A))}{k}  \notag \\
%&\leq \max_{i\in\{1,\ldots,k\}} \mathbb{P}(X_v = x_i|\mathbf{S},\Gamma(V_A)) \notag \\
%&\leq \max_{x\in\mathcal{X}} \mathbb{P}(X_v = x|\mathbf{S},\Gamma(V_A)). \label{eq: cor upp bound}
%\end{align}
%Since the above Equation~\eqref{eq: cor upp bound} holds for any mapping $B$,  it must hold for the optimal mapping as well. 
%\end{proof}

Computing the probabilities in Corollary \ref{cor: opt upper bound} may be challenging, depending on how much information the adversary has. 
Nonetheless, if the adversary can approximate these probabilities with some accuracy (e.g., if it knows the underlying graph $G$), there exist polynomial-time algorithms for computing max-weight matchings \cite{galil1986efficient,bayati2008max}.

The precision-optimal maximum-weight matching does not necessarily maximize recall.
Notice that for any matching, its precision and recall are equal, due to the definitions of precision and recall.
%Moreover, we saw in Theorem \ref{thm:bounds_prec_rec} that $\mathbf R_{\texttt{M}} \geq \mathbf D_{\texttt{M}}$, for all $\texttt{M}$.
%The matching estimator that optimizes precision inherently has the \emph{minimum} possible recall for its precision level.
%While this in itself does not imply anything, it suggests that recall may be maximized by a different estimator than the precision-optimal matching estimator. 
The following theorem characterizes a recall-optimal estimator, which assigns each message $x$ to any server $v$ for which $\prob(X_v = x|\mathbf{O})$ is maximized.

\begin{thm}[Recall-Optimal Estimator] \label{thm: optimal estimator recall}
The \\ recall-optimizing estimator for an adversary with observations $\mathbf{O}=(\mathbf{S},\mathbf{\Gamma})$, is a mapping that assigns each transaction $x\in \mathcal X$ to any server  $v^* \in \argmax_{v\in V_H} \mathbb{P}(X_v = x|\mathbf{O})$. 
\end{thm}
(Proof in Section \ref{proof: optimal estimator recall})

The first-spy estimator is an instance of a recall-optimal estimator for spreading models in which the exit node to the first-spy is the most likely source. 
Moreover, Theorem \ref{thm: optimal estimator recall} implies that a precision-optimal, maximum-weight matching is only recall-optimal if it also maps each message to its most likely source, elementwise.
For example, if  $k$ servers are equally likely sources for $k$ messages, then the precision-optimal matching estimator is also recall-optimal.

\vspace{0.1in}
\noindent \textbf{Summary.} This section provides fundamental limits on both precision and recall, as well as detailing estimators that optimize precision (Thm \ref{thm: optimal estimator}) and recall (Thm. \ref{thm: optimal estimator recall}), respectively.
These fundamental limits and estimators will be useful benchmarks as we analyze the precision-recall regions for networking policies in Sections \ref{sec:algos} and \ref{sec:main}.

%% file: algorithm.tex
\section{Baseline Algorithms}
\label{sec:algos}

With the fundamental bounds from Section \ref{sec:bounds}, we now tackle our main problem: 
designing a networking policy with a minimal detection region.
A key message of our work is that statistical anonymity requires \emph{mixing} of messages: 
users should spread their own messages and those of their peers in a way that is difficult to distinguish.
Degree of mixing depends on three key properties of a networking policy: (1) the spreading protocol, (2) the topology of the network, and (3) the dynamicity of the network (i.e., how often the P2P graph changes). 
For example, the current Bitcoin network uses diffusion spreading over a static, roughly 16-regular topology.
This policy has poor mixing---i.e., a large detection region---because different nodes have unique spreading patterns and can therefore be deanonymized. % more accurately. 
%Offhand, it might seem like rapidly changing the graph solves the problem.
%This strategy turns out to have higher precision than we would like (not to mention a number of practical challenges).
%We specify a networking policy by its (topology, dynamicity, spreading protocol) triple. 
%Surprisingly, it is not enough to simply 

In this section, we first identify a taxonomy of networking policies, based on the  properties above. 
We then systematically evaluate the anonymity of various first-order, natural networking policies from this taxonomy. 
%These baseline policies might reflect a first pass at the problem. 
We show that most of these baseline policies have poor anonymity guarantees, and we extract rules of thumb for improving a policy's anonymity. 
These rules of thumb will build the groundwork for our main result, \Algo, presented in Section \ref{sec:main}.
%Finally, we present a protocol called \Algo-1, which achieves nearly-optimal anonymity guarantees using a simple networking policy.
%For example, it is not enough to simply change the graph every time a transaction is sent. 
%We explain 

\subsection{Taxonomy of Networking Policies}
Our taxonomy has three axes: spreading protocol, topology, and dynamicity. 
We consider multiple categories along each axis.

\noindent \textbf{Spreading protocol.}
The space of spreading protocols is vast. 
In this work, we consider a few natural, first-order spreading policies, and also propose a new protocol called \algo~spreading. 

Perhaps the most natural spreading strategy is \textbf{flooding}, where messages are propagated with a fixed delay to all neighbors. 
A slightly refined version is \textbf{diffusion}, which adds independent randomness to the transmission delays of flooding.
Diffusion is explained in Section \ref{sec:trans_model}.
Flooding and diffusion reflect the current status quo in the Bitcoin network. %, and are known to have poor anonymity properties.

Given that our goal is to provide  anonymity, another natural strategy is to forward a message to a randomly-chosen node, which then runs diffusion or flooding. 
We call this spreading protocol \textbf{diffusion-by-proxy}.

%two spreading protocols: diffusion (with flooding as a special case) and \algo~spreading.
%Diffusion floods content with independent, random delays; it is explained in Section \ref{sec:trans_model}. 
%Flooding is an instance of diffusion where delays are deterministic.

Finally, we propose in this paper a new protocol called \textbf{\algo~spreading}.
Dandelion spreading forwards each message on a randomly-selected line before diffusing it to the rest of the network.
Since \algo~spreading is a comparatively new protocol (not a first-order baseline), we defer a detailed discussion to Section \ref{sec:main}.

\vspace{0.05in}
\noindent \textbf{Topology.} We are interested in topologies that are simultaneously simple to construct, analyzable, and good for anonymity. 
We therefore limit ourselves to a set of canonical graph models: lines, trees, $d$-regular graphs, and complete graphs. 
These categories are not mutually exclusive;
lines are a special case of both trees and regular graphs (we consider lines and cycles interchangeably), and complete graphs are a special case of regular graphs.
%16-regular graphs are a good approximation for the Bitcoin network in practice.

\vspace{0.05in}
\noindent \textbf{Dynamicity.} Many network-based deanonymization attacks use partial or full knowledge of the connectivity graph between nodes \cite{SZ11a}.
We assume that the network can change the graph at varying rates to control the adversary's ability to learn it. 
We consider two extremes on this spectrum: static graphs and dynamic graphs.
In static graphs, the network never changes the graph, so the adversary learns it fully over time.
In dynamic graphs, the graph is changed at a rate such that the adversary only knows its local neighborhood at any given point in time.

\vspace{0.1in}

%\subsection{First-Order Solutions}
\noindent In the remainder of this section, we first explore the regions of our taxonomy by studying three baseline networking policies: flooding, diffusion, and diffusion-by-proxy.
Although none of these baselines has satisfactory anonymity guarantees, the associated analysis provides valuable intuition that helps us design better policies in Section \ref{sec:main}.
%The discussion is organized as in Table \ref{tab:taxonomy}.
%We eventually converge on a theoretically quasi-optimal networking policy called \Algo-1. 
%In Section \ref{sec:systems}, we present \Algo-2, a slight modification that has many of the same anonymity properties, 
%in addition to robustness properties desirable for practical deployment.
%\begin{table}[ht]
%\centering
%\caption{Taxonomy of networking policies. The blue shaded region is the main result of this paper. }
%\label{tab:taxonomy}
%\begin{tabular}{|c|c|c|c|c|}
%\hline
%\multirow{2}{*}{} & \multicolumn{2}{c|}{Diffusion} & \multicolumn{2}{c|}{Dandelion} \\ \cline{2-5} 
%                  & \textit{Static}             & \textit{Dynamic}            & \textit{Static}              & \textit{Dynamic}  \\ \hline
%Line              & \multicolumn{2}{c|}{\multirow{4}{*}{\S \ref{sec:diff}}}  & \multirow{2}{*}{\S \ref{sec:static_trees}}   &  \tiny{\S \ref{sec:dynamic lines}  \Algo-1}       \\ \cline{1-1} \cline{5-5} 
%Tree              & \multicolumn{2}{c|}{}                   &                     &     \S \ref{sec:dynamic trees}     \\ \cline{1-1} \cline{4-5} 
%d-Regular         & \multicolumn{2}{c|}{}                   &                     &    {\cellcolor{blue!25}}   \tiny{\S \ref{sec:systems}, \Algo-2}  \\ \cline{1-1} \cline{5-5} 
%Complete          & \multicolumn{2}{c|}{}                   & \multicolumn{2}{l|}{\S \ref{sec:static_regular}}          \\ \hline
%\end{tabular}
%\end{table}

\subsection{Flooding} 
\label{sec: trees}
To model flooding, we assume that messages propagate  along each graph edge with a deterministic delay, 
and nodes forward incoming messages to their neighbors with a constant delay. 
On undirected topologies, flooding has poor source-hiding due to symmetry and the deterministic spreading scheme.
However, it is unclear if directed topologies fare better.
We begin by showing that flooding has poor performance on directed, static, $d$-regular graphs.
\begin{prop} \label{thm:flooding static regular}
The expected precision of flooding on a static $d$-regular graph is at least $\mathbf{D}_\mathtt{OPT} \geq (1-(1-p)^d)\geq p$.
\end{prop}
(Proof in Section \ref{proof:flooding static regular})

Flooding performs poorly on static regular graphs because each honest node has a unique spreading ``timestamp signature", and the adversary can predict these signatures. 
That is, if node $v$ is the source, then the adversarial nodes receive all messages from $v$ in a deterministic timing pattern. 
Moreover, the adversary can \emph{predict} this pattern from the structure of the graph, due to the fixed nature of flooding.

This reasoning suggests that if the adversary does not know the graph, it cannot predict nodes' spreading patterns, and therefore cannot deanonymize nodes.
However, the following proposition shows that even when the graph is dynamic, the adversary can achieve a high precision. 
\begin{prop} \label{thm:flooding dynamic regular}
Flooding precision for dynamic $d$-regular graphs is bounded as $\mathbf{D}_\mathtt{OPT} \geq cp$ for some constant $c>0$ independent of $p$. 
\end{prop}
(Proof sketch in Section \ref{proof:flooding dynamic regular})

%- assume that time taken for messages to propagate along edges is the same across all edges in the graph; and 
%
%- nodes forward incoming messages to their neighbors with a constant delay. 

This result highlights that even if the adversary cannot predict the exact timestamp pattern for a given node, it can infer certain statistical properties of the pattern that are sufficient for deanonymization. %suggest that flooding is a poor choice for anonymity. 
In short, as long as the topology allows messages to flood in more than one direction, the adversary can use the statistics of observed timestamp signatures to infer the source of a message.

\vspace{0.05in}
\noindent \textbf{Lesson:} Do not flood content in multiple directions on the graph at the same rate.

\subsection{Diffusion}
\label{sec:diff}
Diffusion is a natural successor to flooding; instead of using deterministic delays, it uses random ones.
%In principle, this should give better anonymity than flooding. 
By introducing uncertainty into the adversary's timing estimates, diffusion reduces the adversary's overall precision and recall.
However, much research in recent years has shown that the source of a diffusion process can nonetheless be identified reliably \cite{SZ11a,SZ12,FC12,WDZT14,LMOZ13,PVF12,PTV12,ZY13}. 
Although there are no theoretical results on the precision or recall under our particular adversarial model, 
several heuristic estimation algorithms
are able to identify the source of a diffusion process on many classes of graphs \cite{PTV12,ZY13}. 
Moreover, theoretical results exist on other adversarial models \cite{SZ11a,SZ12,WDZT14}.
All of these results rely on the intuition that diffusion spreads content symmetrically. 
Because of this, the source node appears at the center of the adversary's observed spreading pattern, and can be identified. 
Diffusion is therefore not a satisfactory solution to this problem.

\vspace{0.05in}
\noindent \textbf{Lesson:} Random forwarding delays are not powerful enough to provide anonymity against spreading protocols that spread content symmetrically. 

\subsection{Diffusion-by-Proxy}
\label{sec:diff_proxy}

The takeaway message from diffusion and flooding is that symmetry of spreading leads to deanonymization. 
To counter this, we must break the symmetry of diffusion. 
A natural strategy for breaking symmetry about the source is to ask someone else to spread the message.
That is, for every transaction, the source node chooses a peer uniformly at random from the pool of all nodes.
It transmits the message to that node, who then broadcasts the message. 
More generally, the network could forward each message a few hops (each hop choosing a new node at random) before diffusing it.
We call this approach diffusion-by-proxy, and it is conceptually equivalent to propagating over a line that changes for every transmission. 
Diffusion-by-proxy might seem like it should have low precision because the graph is so dynamic, but that intuition turns out to be false.

\begin{prop} \label{thm:complete}
The expected first-spy precision of diffusion-by-proxy is bounded as 
\begin{align}
\mathbf{D}_\mathtt{FS} \geq \frac{p}{1-p}(1-e^{p-1}). %\tilde{n}.
\end{align}
\end{prop}
(Proof in Section \ref{proof:complete})

Intuitively, this statement holds because each node delivers its own message to the adversary with probability $p$, and few other nodes report to the adversary over the same edge. 
%Therefore, conditioned on the proxy being a spythe first-spy estimator identifies the true source with a precision 1, 
%and there are 
So even though diffusion-by-proxy breaks the symmetry of diffusion, 
it also provides many paths for messages to reach the adversary. 
%The random-proxy-selection steps can be thought of as choosing an edge from a complete graph, which means that \emph{any} honest node can deliver a message to the adversary. 
Since there are many total paths to the adversary, each path sees (relatively) less traffic, which in turn reduces the amount of mixing that happens. 
A simple countermeasure is to reduce the number of paths over which messages can flow. 

\vspace{0.05in}
\noindent \textbf{Lesson:} There is anonymity in numbers; dense graphs achieve poor mixing because they do not constrain messages to flow over the same paths.

%% file: main.tex
\section{Main Result: Dandelion}
\label{sec:main}
The baseline spreading protocols from Section \ref{sec:algos} provide us with a key  guideline for building more anonymous networking policies: spread asymmetrically over a sparse graph.
In this vein, we propose a new protocol: \textbf{dandelion spreading}. 
While the basic intuition of \algo~spreading is used in several point-to-point anonymous communication systems \cite{tor,reiter1998crowds}, 
it has not been formally studied in the context of anonymous broadcast messaging. 

\begin{algorithm}
\DontPrintSemicolon
\KwIn{Message $X_v$, source $v$, anonymity graph $G$, spreading graph $H$, parameter $q\in (0,1)$}
%\KwOut{A connected, directed graph $G(V,E)$ with average degree $d$}
anonPhase $\gets$ True \;
head $\gets v$ \;
recipients $\gets  \{v \}$ \;
\While{anonPhase} {
    \tcc{forward message to random node} 
    target $\sim$ Unif$(\mathcal N_{out}(G, \text{head}))$\;
    recipients $\gets$  recipients $\cup \{X_v\}$ from head to target \;
    head $\gets$ target \;
    $u \sim$ Unif$([0,1])$ \;
    \If{$u \leq q$} {
      anonPhase $\gets$ False \;
    }
}
\tcc{Run diffusion over $H$ from `head'} 
{\sc Diffusion}$(X_v, \text{head}, H)$
\caption{{\sc Dandelion Spreading}. $\mathcal N_{out}(G,v)$ denotes the out-neighbors of node $v$ on directed graph $G$.}
\label{algo:dandelion}
\end{algorithm}

Dandelion spreading consists of an anonymity phase and a spreading phase (Algorithm \ref{algo:dandelion}). 
In the anonymity phase, the protocol spreads the message over a randomly-selected line for a random number of hops;
in the spreading phase, the message is broadcast using diffusion until the whole network receives the message.
In general, the two phases can occur over different graphs. 
In this work, we will design a (possibly time-varying) graph $G$ over which the anonymity phase occurs, and we will assume the spreading phase occurs over the current Bitcoin P2P network $H$.
The name `\algo~spreading' reflects the spreading pattern's resemblance to a dandelion seed head (Figure \ref{fig:dandelion}).

\begin{figure}[h]
    \centering
  \includegraphics[width=.3\textwidth]{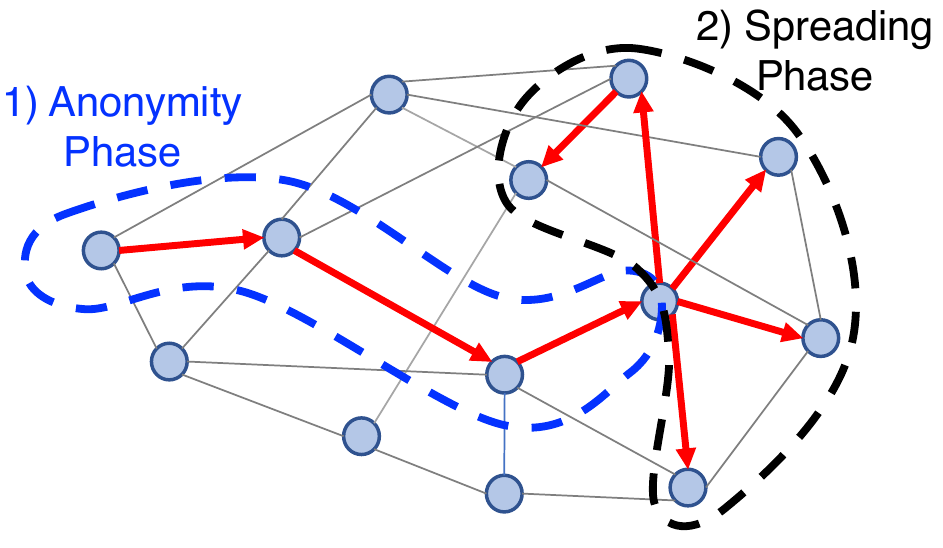}
  \caption{Dandelion spreading forwards a message in a line over the graph, then broadcasts it using diffusion. Here both phases occur over the same graph, i.e., $H=G$.}
  \label{fig:dandelion}
\end{figure}

The two-phase nature of \algo~spreading  allows us to separately  design networking policies that optimize anonymity and latency.
This separated architecture is not necessarily optimal in terms of a latency-anonymity tradeoff; exploring that tradeoff is an interesting direction for future work.
However, diffusion is known to have good spreading properties \cite{bartlett1956deterministic}, but poor anonymity properties \cite{SZ11a}.
Therefore, we combine it with an anonymity phase of constant duration (in an order sense), such that the average latency is increased by a small, bounded factor.
We subsequently assume that the spreading phase can be fully deanonymized; 
i.e., the node that launches the diffusion process can be identified.
As such, we only need to analyze the precision and recall of the anonymity phase.
This assumption does not weaken our anonymity guarantees since it gives the adversary more power.

%Recall that we only need to analyze the anonymity phase, since the spreading phase is assumed to provide no anonymity.

A key observation for this analysis is that the anonymity phase of \algo~spreading largely \emph{removes the need for exact timestamps}.
For honest server $v$, let $S'_v\subseteq S_v$ denote a trimmed down version of $S_v$, in which we retain only those transaction log tuples $(x,u,T_u(x))$ that correspond to the first time transaction $x$ was received by an adversary from any honest node. % in the network. 
That is, we only keep a tuple if $u$ was the first spy to see message $x$, and $x$ was delivered to $u$ by honest exit node $v$.
As before, let $\mathbf{S}'$ denote the vector of all $S'_v$'s. 
Then, with \algo~spreading, it holds that $\mathbf{X} - (\mathbf{S}',\mathbf{\Gamma}) - \mathbf{S}$ forms a Markov chain. 
Therefore it is sufficient to use only the first observation information $\mathbf{S}'$ instead of $\mathbf{S}$ for computing transaction likelihoods. 
In fact, the sufficient statistic $\mathbf{S}'$ can be further simplified by ignoring the timestamp coordinate $T_u(x)$ in the tuples. 
This is possible due to our assumption that the transactions' originating times  are  unknown \emph{a priori} to the adversary, which removes the observed timestamps from any temporal reference frame.
Hence, in the remainder of this paper, with a slight abuse of notation, we use $S_v$, for honest server $v$, to denote the set of message tuples $(x,u)$ such that (i) $u$ was the first adversarial node to receive $x$ and (ii) $u$ received $x$ from $v$. $\mathbf{S}$ denotes the vector of $S_v$'s.

Note that a similar argument does not hold for spreading mechanisms like flooding or diffusion, in which multiple independent timestamps across different nodes (i.e., not just the first observation of a message)  are used to compute likelihoods. 
The diversity of such observations allows the estimator to compare timestamps across nodes, thus making them useful for detection. 
%That is, after the first spy in the anonymity phase receives a message, subsequent spies' observations (including timestamps) are independent of the true source, conditioned on the first spy's information.
%This is true because the message spreads unidirectionally on a line, so a Markov property holds.
%Moreover, we make no assumptions on transaction creation times, 
%and we can use random forwarding delays to decorrelate the adversary's observed timestamps from the transaction's time of creation.
%% we do not model the relay delay time. 
%%For a given message, the first spy's timestamp therefore reveals no information about the true source's identity, or the time at which the message was generated.
%%Even if we were to model the delay of relaying a message, 
%This implies that only two pieces of useful information can be recovered from the observed timestamps: the identity of the first spy, and the exit node that delivers a message to the first spy.
%Our analysis of \algo~spreading therefore assumes the adversary's estimator uses: (1) the set $\mathbf S$ of messages reported to the adversary by each honest node, and (2) the subgraph of $G$ that it has learned.

\vspace{0.01in}
We begin by showing that the maximum recall for \algo~spreading over any connected topology is $p$, the lower bound from Theorem \ref{thm:lower_bounds_fs}.
\begin{thm} \label{thm:dandelion recall}
The expected maximum recall for \algo~ \\ spreading on any connected graph of $n$ nodes with a fraction $p$ of adversaries is $\mathbf{R}_\mathtt{OPT} = p$. %\tilde{n}p$. 
\end{thm}
(Proof in Section \ref{proof:dandelion recall})

The reason for this result is that \algo~spreading propagates content unidirectionally over a line. 
This lack of symmetry makes the first-spy estimator---which has a recall of $p$---optimal.
Theorem \ref{thm:dandelion recall} result implies that as we explore various topologies of \algo~spreading, we only need to analyze and minimize their precision.
We do so for three topologies of the graph $G$: static trees, dynamic trees, and dynamic lines.
Each topology provides intuition about how to achieve anonymity.
We find that dynamic lines achieve nearly-optimal average precision and recall.

\subsection{Static Trees}
\label{sec:static_trees}

Recall that our goal is to mix messages from different users; in this sense, trees are a natural topology to study.
That is, consider a rooted, directed $d$-regular tree, with each edge directed toward the parent node. 
Dandelion spreading respects the directedness of the graph, so during the anonymity phase, each node passes all messages to its parent node (i.e. toward the root). 
Nodes near the root are therefore able to mix their own messages with exponentially-many other messages from users beneath them in the tree.
However nodes near the leaves of the tree have few nodes beneath them, and therefore experience minimal mixing. 
%For example, if a leaf node has an adversarial parent, then the precision and recall for that leaf are both one.
This fundamental asymmetry results in a high average precision. 

\begin{prop} \label{thm:static tree}
The expected precision under a matching estimator $\mathtt{MAT}$ on any tree is given by $\mathbf{D}_\mathtt{MAT} = p$. %\tilde{n}p$. 
\end{prop}
(Proof in Section \ref{proof:static tree})

Intuitively, when the graph is known, the adversary can partition nodes into \emph{wards}, or sets of honest nodes that share the same first spy.
Each ward contributes equally to the adversary's precision, so we would like to minimize the number of wards.
On trees, the expected number of wards is $p\tilde n$, most of which consist of a single leaf with an adversarial parent node.
This gives an overall precision of $p$.

Although a precision of $p$ is an improvement over Bitcoin's current networking policy, we would like to achieve a precision close to the lower bound of $p^2$ (Theorem \ref{thm:lower_bounds_fs}).
We therefore consider topologies with fewer wards on average. 
%Specifically, we seek topologies with similar mixing to trees, but without the high number of leaves.
%Therefore, we move to a topology that exhibits similar mixing to trees, without the high number of leaves.

\vspace{0.05in}
\noindent \textbf{Lesson:} 
Use topologies in which it is difficult for the adversary to partition nodes into wards.

\subsection{Dynamic Trees}
\label{sec:dynamic trees}
The adversary was able to partition the nodes of a static tree into wards largely because the graph was known.
A natural question is whether dynamic trees have the same problem, since most of the graph is hidden, except the adversary's local neighborhood.
%We have seen that for static trees, expected precision is $p$.  
%We now explore performance in the dynamic case, i.e., when the graph is unknown to the adversary, except for its own local neighborhood.
%This can be achieved by changing $G$ quickly enough to prevent the adversary from learning it.

A perfect $d$-ary tree is a rooted tree in which each node has either $d$ children or no children, and all leaves are at the same depth.
Again, we assume each edge in such a tree is directed toward the parent node.
We find that \algo~spreading on perfect $d$-ary trees has an expected precision similar to that of static trees.
%- where topology is only locally known to adversaries (keep changing frequently so that they cannot learn) \\

%\noindent {\bf Trees:}
%
%- for completeness let's start with trees. 
%
%- unsurprisingly they don't do well, for reasons similar to static trees. 
\begin{prop} \label{thm: dynamic trees}
The expected precision of the first-spy estimator on a % balanced
perfect $d$-ary tree, $d\geq 2$, can be bounded as $\mathbf{D}_\mathtt{FS} \geq p/2$.  
\end{prop}
(Proof in Section \ref{proof: dynamic trees})

Since the graph is now dynamic, the adversary cannot explicitly determine every ward like it could in the static case. 
However, the first-spy estimator naturally identifies wards that consist of a single honest leaf. 
Statistically, there are many such wards on trees that are not lines, so we obtain similar guarantees to the static case.
This implies that the problem with trees is mainly the fact that they have many leaves.

\vspace{0.05in}
\noindent \textbf{Lesson:} A dynamic graph does not mitigate the negative impact of leaf nodes.

\subsection {Dynamic Lines: Dandelion}
\label{sec:dynamic lines}

Next, we study dynamic line graphs.
Lines are 2-regular trees, but unlike higher degree trees, they do not suffer from the asymmetry problems associated with leaves. 
However, line graphs seem to lack the strong mixing properties of higher-degree graphs.
Nonetheless, we show near-optimal precision for this class of graphs.
This happens because despite the moderate mixing on lines, the number of honest nodes visible to the adversary is also small.
As such, the adversary cannot accurately partition nodes into wards, which reduces the overall precision.
Note that this  would not hold in the static case, since the adversary could identify the wards exactly.

\begin{figure}[h]
    \centering
  \includegraphics[width=.35\textwidth]{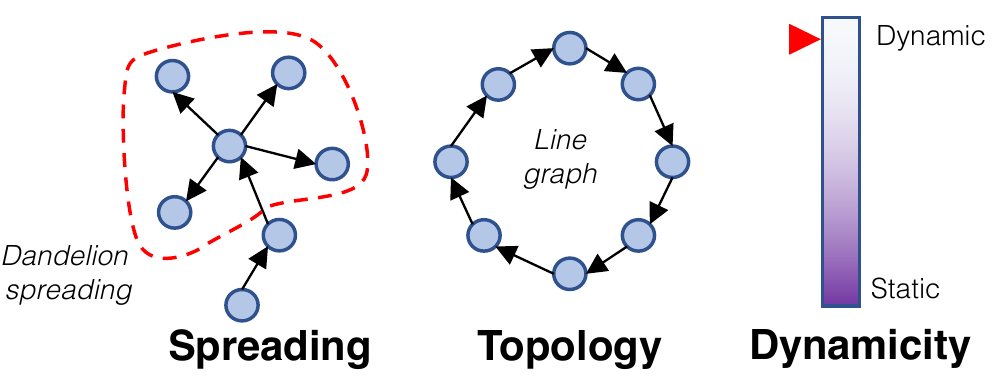}
  \caption{The \Algo~networking policy: (1) \algo~spreading, (2) a line topology, (3) a dynamic graph.}
  \label{fig:dandelion_policy}
\end{figure}

We use the name \textbf{\Algo} to refer to a full networking policy (Figure \ref{fig:dandelion_policy}): \algo~spreading over dynamic lines (i.e., $2$-regular graphs with out-degree 1).
We begin by showing that \Algo~ has near-optimal precision.
%- however a priori, the line graph seems to be lacking in the strong mixing capabilities as guaranteed by a higher degree tree for e.g.. 

%- nevertheless, in this case, we can show near-optimal performance! 

%- this is because, though the mixing is only moderate, the number of honest nodes visible to the adversary is also small. so together these effects yield a near-optimal payoff. 

%- note that this strategy does not work in the static case since the visibility benefit is bereft. 

\begin{thm} \label{thm: dynamic line}
The expected precision of \Algo~(i.e., \algo~spreading on a dynamic line graph) with $n$ nodes and a fraction $p<1/3$ of adversaries, is upper bounded by 
\begin{align}
\mathbf{D}_\mathtt{OPT} \leq \frac{2p^2}{1-p}\log\left( \frac{2}{p} \right) + O\left (\frac{1}{n}\right ). 
\end{align}
\end{thm}
(Proof in Section \ref{proof: dynamic line})

This result states that for small $p$, the expected maximum precision is within a logarithmic factor of our lower bound of $p^2$.
The stated bound has loose constants for improved readability;
a  tighter expression is included in the proof. 
The proof depends heavily on the fact that the adversary cannot reliably assign nodes to wards outside of its local neighborhood on the graph. 
As such, it is forced to use estimators that would give suboptimal precision in the static case, 
like variants of the first-spy estimator.

Figure \ref{fig:region_final} illustrates \Algo's detection region compared to those of other benchmark policies.
The points for diffusion and flooding are generated through simulation over a snapshot of the Bitcoin server graph from 2015 \cite{coinscope}. 
Since \algo~spreading has optimally-low recall (Theorem \ref{thm:dandelion recall}),  the Pareto frontier for \Algo~is exactly the plotted point. 
The other policies are analyzed using possibly-suboptimal estimators, so their detection regions must at least contain the plotted points.

\begin{figure}[t]   \label{fig:region_final}
    \centering
  \includegraphics[width=.41\textwidth]{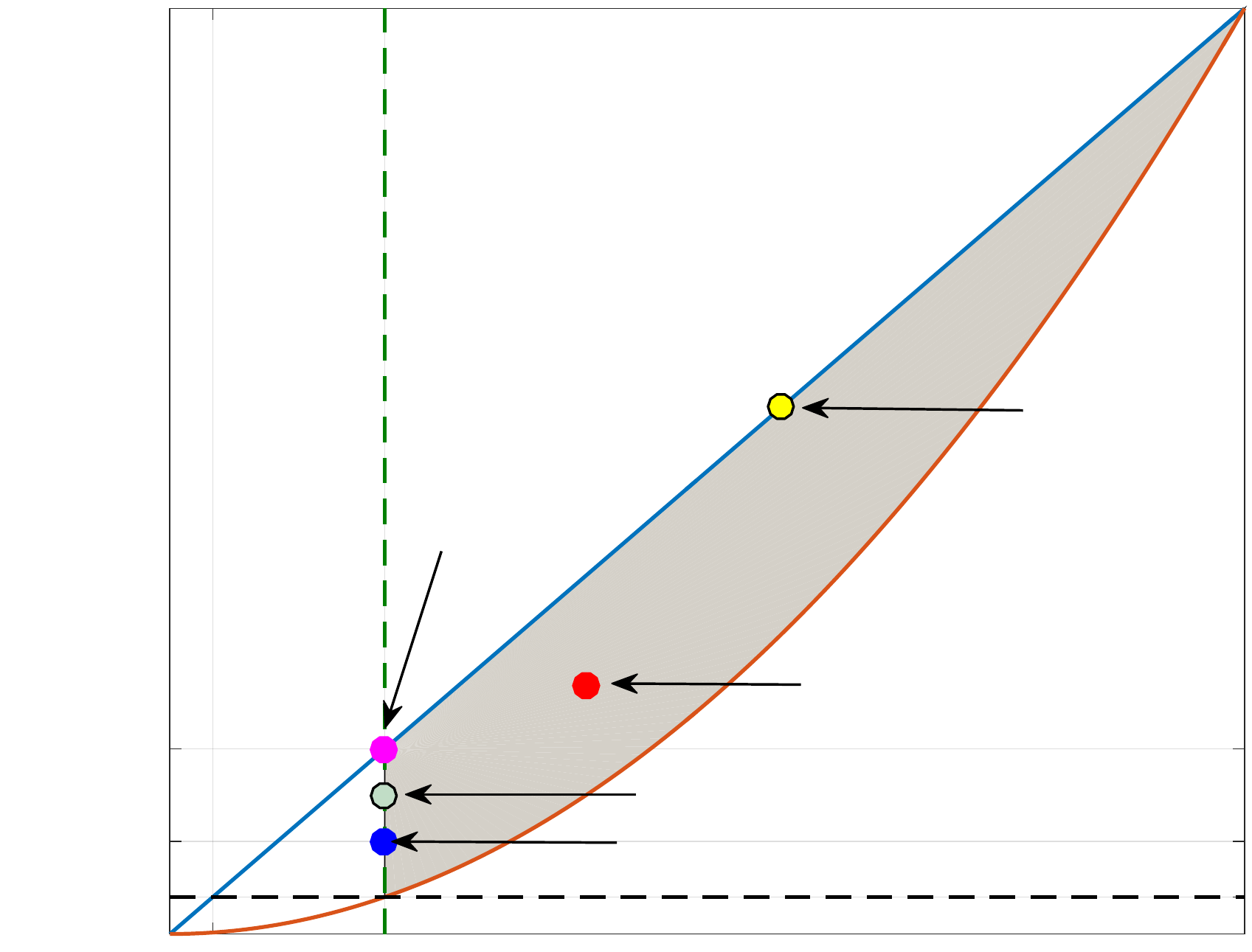}
  \put(-96,-15){Recall}
  \put(-184,-8){0}
  \put(-172,-8){$p^2$}
  \put(-142,-8){$p$}
  \put(-100,13){\Algo}
  \put(-100,23){Dynamic Trees}
  \put(-70,41){Diffusion}
  \put(-150,70){Static Trees}
  \put(-38,85){Flooding}
  \put(-2,-8){$1$}
  % vertical
  \put(-188,3){$p^2$}
  \put(-221,15){\rotatebox{0}{$\frac{2p^2}{1-p}\log(\frac{2}{p})$}}
  \put(-188,30){$p$}
  \put(-188,146){$1$}
  \put(-200,55){\rotatebox{90}{Precision}}
  \caption{Detection regions for studied networking policies, $p=0.2$. 
 \Algo~has a detection region close to the fundamental lower bounds. }     \label{fig:region_final}
\end{figure}
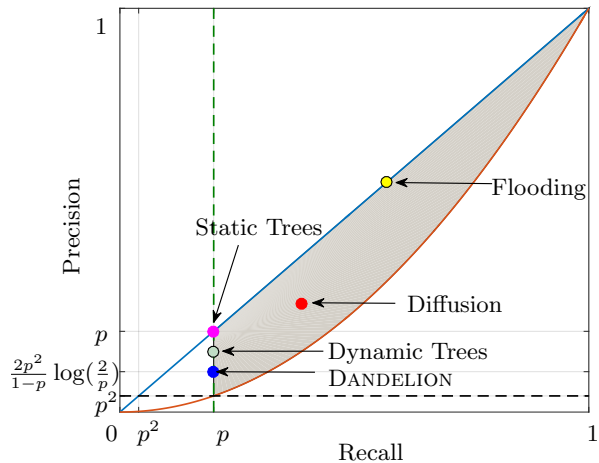

\Algo~ satisfies the theoretical demands of our problem.
However, implementing \Algo~in a distributed, robust fashion is nontrivial. %presents some practical challenges related to distributed graph maintenance. % suffers from a number of robustness issues that may hamper its performance in a real system. 
In the next section, we discuss some challenges associated with implementing \Algo~ and present a simple, distributed implementation with empirically good performance.
%\Algo++: a networking policy with similar anonymity performance to \Algo~, but improved robustness.

%\noindent {\bf Regular graphs:} 
%
%- moving on to regular graphs. 

%- though each node see a large subtree below it, in effect, the average number of messages that mix at a node is roughly the same as that on a line graph (~1/p).
%
%- on the other hand visibility is increased on a regular graph, due to its increased number of edges
%
%- hence overall we expect a higher payoff. 
%
%- again super-painful to analyze, but with luck can work out some bounds.  

%% file: systems.tex
\section{Systems Issues}
\label{sec:systems}
Theoretically, \Algo~ is simple and exhibits desirable anonymity properties.
However, the implementation raises a number of practical considerations, like how to construct the underlying line graph and how to provide sufficient graph dynamicity.
%However, a number of practical considerations, like distributed graph construction and graph leakage, prevent \Algo~from being a robust solution for real cryptocurrencies.
%In this section, we present \Algo++:  a simple protocol that closely approximates the anonymity performance of \Algo, while providing robustness to many of the associated practical challenges.
We discuss each of these challenges, and introduce practical, heuristic solutions for addressing them.

%\subsection{Challenges}
\subsection{Constructing a line graph}
%\red{- section number should be 5.0.1?} 

In \Algo, all nodes propagate their messages over the \emph{same} line.
To implement this, the network must build either a Hamiltonian circuit or a set of long, disjoint lines in a fully-distributed fashion.
%(we refer to line and cycle graphs interchangeably).
Constructing a Hamiltonian circuit is challenging in our case because it is not a one-time event; in order to provide dynamicity, the network must frequently construct a new random line. 
To ensure scalability, the algorithm for constructing such a line should be fully-distributed, lightweight, and asynchronous.
%These constraints arise because , so structured coordination is unlikely to be a scalable solution. 
%For instance, consider the following first-order protocol:
%%On the other hand, first-order methods for generating lines or line segments typically require sequential construction and/or incur significant communication costs.
%one node starts as the `head' of the line, and the rest of the nodes are in a pool of unconnected nodes.
%In each iteration, the head connects to a random node from the unconnected pool, passing it the head token.
%The new head is removed from the unconnected pool.
%Eventually this protocol creates a line, but keeping track of the unconnected nodes without significant communication or access to centralized resources is nontrivial.
%Moreover, the nodes must act in a sequential order, which may be difficult to enforce in practice.

%These problems are common to sever
Traditional algorithms for computing Hamiltonian circuits are often computationally intensive and/or require centralized control 
\cite{garey1976planar,karger1997approximating},
but recent papers have studied lightweight, distributed alternatives  \cite{srikant,levy2004distributed}.
For instance, \cite{levy2004distributed} first generates line fragments, then splices them together into a circuit. 
However, it relies on the nodes of each line fragment knowing the identities of the fragment's head and tail nodes.
This could partially reveal  the graph structure to the adversary, which would likely change our anonymity guarantees.

On the other hand, \cite{srikant} builds up the circuit sequentially;  a pair of nodes start as the circuit `seeds'. %and  nodes iteratively join the circuit by splicing an edge uniformly at random.
Each node $v$ who joins the circuit contacts a random node $u$ from the partially-built circuit; 
$u$ replies with the IP address of  its outgoing neighbor $w$.
Then $v$ splices itself into the $(u,w)$ edge, so the new ordering is $u\rightarrow v\rightarrow w$. 
This distributed protocol is a viable solution for constructing an exact line.
%We cannot impose an explicit ordering on nodes during the graph-construction protocol, but the inherent randomness in nodes' timing may naturally impose an ordering on their actions. 
%Although this protocol is sequential, it does not use global information to build the circuit, making it a viable solution. %---nodes simply contact one another uniformly at random.
%%As such, it is a viable solution.
%However, a weakness is that the adversarial nodes could force honest nodes into single-node wards by always returning another spy node as their ``outgoing connection". 
%The honest node would then find itself sandwiched between malicious nodes. 
%Rather, they contact a node at random, and 
%This model makes sense for a P2P network that is being constructed on-the-fly, as new nodes join the network. 
%However, arranging for an existing network to run this protocol in succession may be challenging. 
%However, in order to maintain a dynamic graph $G$, this algorithm needs to be run frequently, on the order of every few minutes. 
%This requirement discounts algorithms that are  %, so it is important to choose an efficient algorithm.
%This renders them unsuitable for our setting.

%\red{- a one-time construction of a Hamiltonian circuit in P2P is pretty easy (e.g. https://arxiv.org/abs/1207.3110): so need to emphasize that the distributed algorithm is for restructuring the network.}

Another alternative is to use Bitcoin's current networking strategy to approximate a line.
Currently, each Bitcoin server generates eight connections at random. 
%In general, this approach can be used to approximate a $d$-regular digraph by asking each server to create $d/2$ outgoing connections uniformly at random. 
We can approximate a line by asking each server to create one outgoing connection at random. 
%For a line, each server would therefore make one connection, giving an expected degree of 2. 
We can further refine this protocol by having each server, prior to making a connection, contact $k$ nodes and connect to the node with the smallest in-degree.
%, prior to making a connection, ask the target node if its in-degree is $\geq d/2$. 
%If so, the server randomly contacts another node. 
%This process can be repeated $k$ times. 
Algorithm \ref{algo:dreg_approx} specifies this algorithm for approximating a line.
%To maintain dynamicity, each node periodically drops all connections and makes new ones according to Algorithm \ref{algo:dreg_approx}.
%As with \cite{srikant}, spy nodes can hijack connections by lying about their degree. 
%However, the adversary cannot force honest nodes into single-node wards, and the problem disappears entirely for $k=1$.
The protocol is fully distributed, but it is unclear how well it approximates a line.
%Its primary advantage over \cite{srikant} is that it is already implemented in the Bitcoin networking stack, and therefore requires no new code.
%\red{- would it be better to write Algorithm~\ref{algo:dreg_approx} from the perspective of an individual node, since it's a distributed algorithm (i.e., get rid of the outer for-loop essentially) ? A: I think it might be slightly more confusing for the reader that way, but am not opposed to it... }

\begin{algorithm}[t]
\DontPrintSemicolon
\KwIn{Set $V=\{v_1, v_2, \ldots, v_n\}$ of nodes; parameter $k$}
\KwOut{A connected, directed graph $G(V,E)$ with average degree $2$}
\For{$v \gets V$} {
%  \For{$c \gets 1$ \textbf{to} $d$} {
    \tcc{pick $k$ random targets} 
    $u_i \sim$ Unif$(V \setminus \{v\}),~~\text{for }i\in \{1,\ldots, k\}$ \;
    \tcc{pick the smallest in-degree}
    $u \gets \argmin_{u_i} \text{deg}_{in}(u_i)$ \;
    $E = E \cup (v \rightarrow u)$ \tcc{make connection} 
%  }
  
}
\Return{$G(V,E)$}\;
\caption{{\sc $k$-Approximate Line} Approximates a directed line graph in a fully-distributed fashion. Each node picks an edge from $k$ options}
\label{algo:dreg_approx}
\end{algorithm}

%\red{- in the steady-state if a node already has some edges, but wants to change the graph structure according to Algorithm~\ref{algo:dreg_approx}, are the existing edges dropped?}

Figure \ref{fig:line_approx} illustrates the degree distribution of Algorithm \ref{algo:dreg_approx}'s approximation of a line graph with 1,000 nodes, averaged over 1,000 trials, for different values of $k$. 
First, note that the average degree is two by construction.
As $k$ increases, the fraction of leaves decays, with the greatest reduction coming as we transition from $k=1$ to $k=2$.
This empirical observation is supported  by the following proposition:
\begin{prop} \label{prop:twochoices}
Suppose Algorithm \ref{algo:dreg_approx} is used to construct a $k$-approximate line over $n$ nodes. 
Let the empirical degree distribution of the resulting graph's nodes have support $(d_1, \ldots, d_m)$, where $d_1 <  \ldots < d_m$. 
Then with probability $1-o(1)$, the maximum degree $d_m$ satisfies the following condition:
\[
d_m =
\begin{cases} 
      \frac{ \log n}{\log \log n} \left(1 + o(1) \right ) + \Theta(1) &\text{if } k= 1 \\
      \frac{\log \log n}{\log k} \left(1 + o(1) \right ) + \Theta(1) & \text{if } k>1. 
   \end{cases}
\]
\end{prop}
(Proof in Section \ref{proof:twochoices})

Here we are using maximum degree as a proxy for regularity (or number of leaves), but recall that the expected degree is fixed by construction. 
Therefore, if we can drive the maximum degree down to $2$, the minimum degree must also be $2$.
Proposition \ref{prop:twochoices} suggests that we can reap most of the precision gains of a more regular graph by connecting to one of $k=2$ nodes with minimum in-degree, whereas larger $k$ only improves the regularity by a factor logarithmic in $k$.

Sec. \ref{sec:dynamic trees} showed that leaves increase the precision of a scheme because the leaf nodes' messages cannot be mixed with other messages.
This suggests that \Algo~can achieve lower precision over $k$-approximate lines (Algorithm \ref{algo:dreg_approx}) by increasing $k$ and decreasing the number of leaves.
Figure \ref{fig:compare_dreg} compares the the first-spy estimator precision for exact lines (generated by \cite{srikant}) and $k$-approximate lines (Algorithm \ref{algo:dreg_approx}).
The figure shows that over $k$-approximate lines, average precision decreases as $k$ increases (i.e., as the distribution becomes more regular),
but the returns are diminishing in $k$.
The most significant decrease in precision occurs as we transition from $k=1$ to $k=2$;  higher values of $k$ give marginal improvements.
Moreover, the precision of $k$-approximate lines is significantly larger than that of exact lines, which could be obtained through the line-creation protocol in \cite{srikant}.
%\red{TODO: Work out the $k=1$ condition from proofs, make sure it's right.}
%Equivalently, the regularity of the graph increases.

\begin{figure}[t]
    \centering
  \includegraphics[width=.32\textwidth]{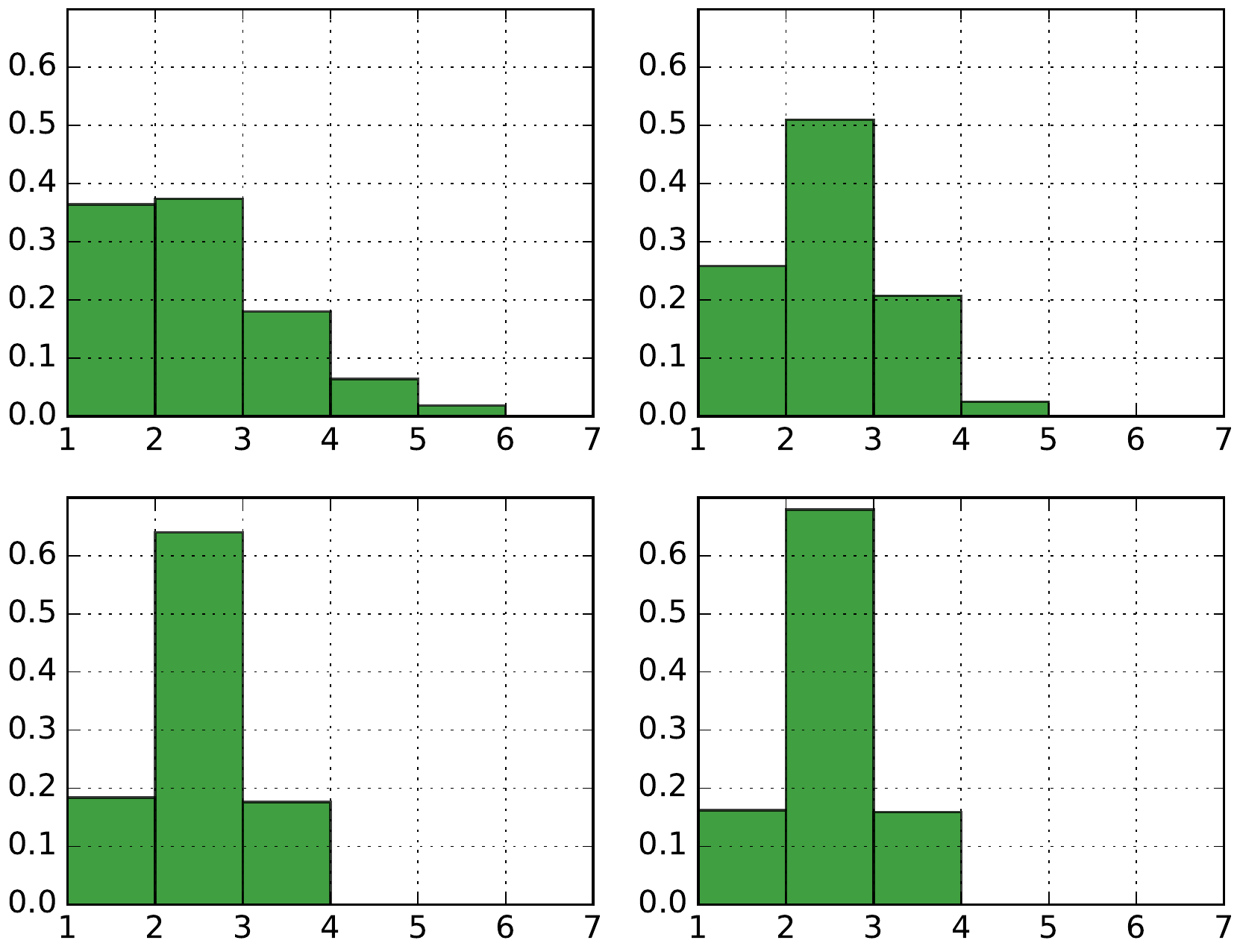}
  \put(-95,-10){Degree}
%  \put(-68,166){$4$-Regular}
%  \put(-173,166){$2$-Regular}
  \put(-35,105){\fcolorbox{black}{yellow}{$k=2$}}
  \put(-35,35){\fcolorbox{black}{yellow}{$k=4$}}
%  \put(-40,30){\fcolorbox{black}{yellow}{$k=2$}}
  \put(-120,105){\fcolorbox{black}{yellow}{$k=1$}}
  \put(-120,35){\fcolorbox{black}{yellow}{$k=3$}}
%  \put(-147,30){\fcolorbox{black}{yellow}{$k=2$}}
  \put(-180,30){\rotatebox{90}{Fraction of nodes}}
  \caption{Degree distribution of $k$-approximate lines (Algorithm \ref{algo:dreg_approx}) for various $k$.
  The fraction of leaves decreases as the number of edge choices $k$ increases.}
  \label{fig:line_approx}
\end{figure}
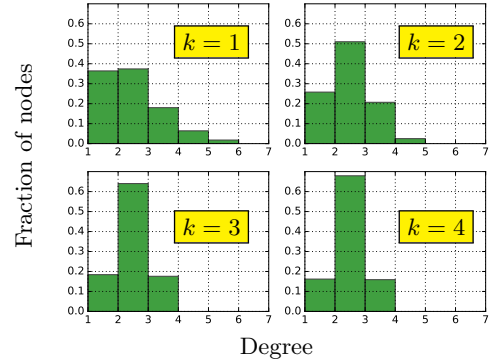

%The main reason for keeping $k$ small is to improve the efficiency of the graph-construction algorithm, 
%which is run regularly to prevent the adversary from learning the graph.

Algorithm \ref{algo:dreg_approx} and \cite{srikant} are both viable options for constructing a line. 
Although \cite{srikant} has lower overall precision, it uses more fine-grained information---connection IPs rather than simple degree information.
As such, \cite{srikant} may be less robust to misbehaving nodes.
Understanding this tradeoff, and developing alternatives that are resistant to adversarial misbehavior, are of practical interest.

%The key observation is that the number of leaf nodes of degree 1 is significant, even as the rounds of degree-checking $k$ increases. 
%Indeed, for $k=0$, a constant fraction of nodes are expected to be leaves:
%\begin{prop}
%The fraction of leaf nodes as a function of $k$  as $n\rightarrow \infty$ when $k=0$. $1/e$
%\end{prop}
%Figure \ref{fig:compare_dreg} confirms this intuition by showing that the average precision of the first-spy estimator on a $k$-approximate line decreases with $k$.
%Nonetheless, the fraction of leaves does not seem to disappear, which causes the overall precision to be higher than that of an exact line.

% can be three times higher than on an exact line (i.e., a $2$-regular digraph). 

%These observations suggest that in practice, it may be difficult to construct exact line graphs in a distributed fashion, and simple approximations are likely to have higher precision than the theoretical lower bounds would suggest. 

\begin{figure}[t]
%     \null \hfill
%     \subfloat{%
       \centering
       \includegraphics[width=0.37\textwidth]{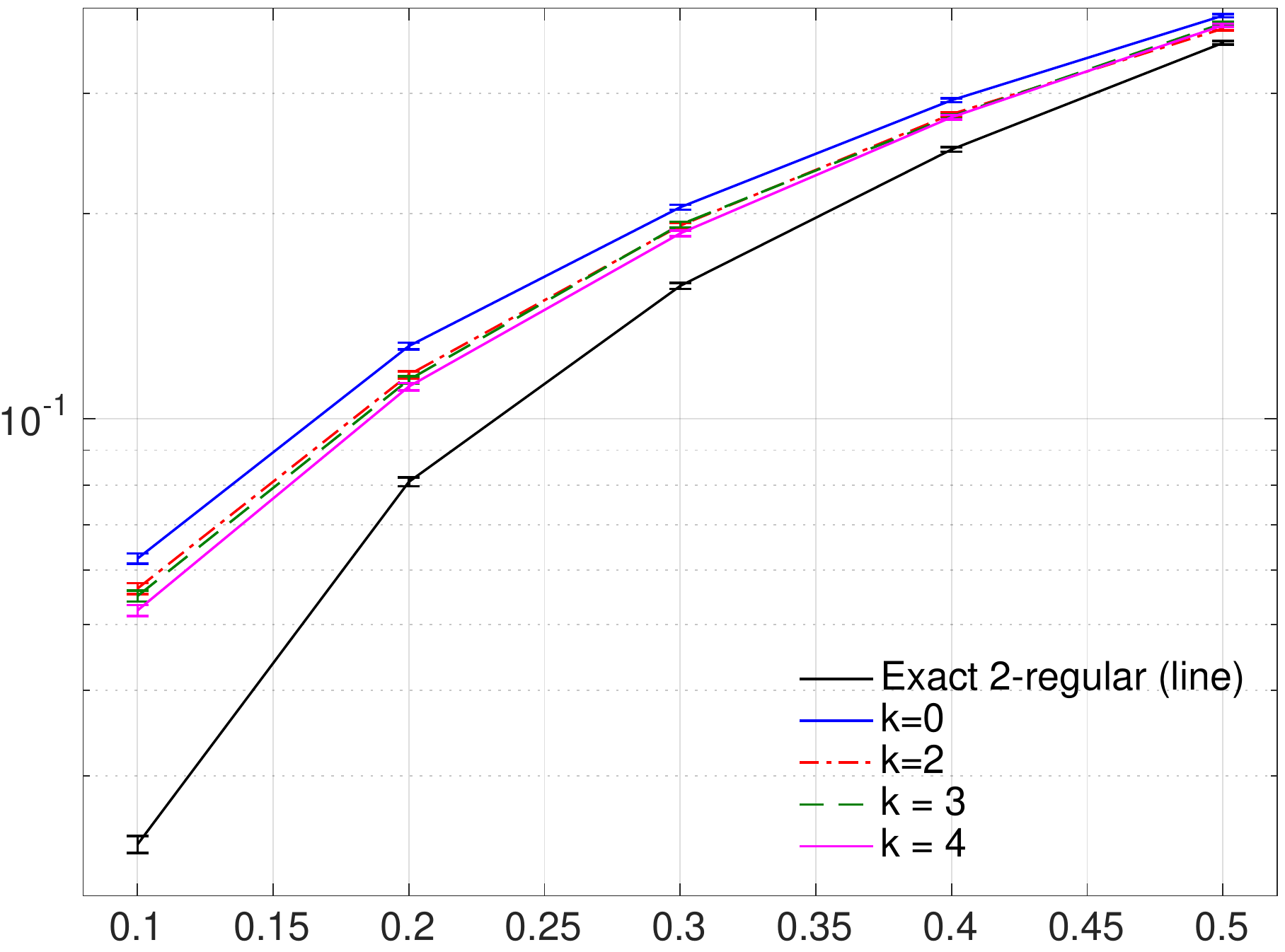}
       \put(-125,-10){Fraction of spies $p$}
       \put(-198,45){\rotatebox{90}{Average precision}}
%       \put(-125, 140){$d=2$ (line graph)}
%     }
%     \hfill
%     \subfloat{%
%       \includegraphics[width=0.37\textwidth]{figures/d4_reg_first_spy_approx}
%       \put(-125,-10){Fraction of spies $p$}
%       \put(-198,45){\rotatebox{90}{Average precision}}
%       \put(-135, 140){$d=4$ ($d$-regular graph \red{<- can remove this})}
%     }
%     \hfill
%     \null 
     \caption{The $k$ edge choices during graph creation (Algorithm \ref{algo:dreg_approx}) do not significantly reduce the precision of the first-spy estimator beyond $k=1$.
%     Moreover, $4$-regular graphs have lower overall precision than lines for $k\in \{0,1,2\}$. 
     }
     \label{fig:compare_dreg}
\end{figure}

\subsection{Preventing graph leakage}
Another challenge associated with \Algo~ is that it assumes the graph $G$ is a line whose structure is unknown to the adversary. 
However, lines can be learned over time.
%In this work, we focus primarily on the anonymity phase, but the spreading phase also leaks information.
First, note that for any given adversarial node $s_1$ on a 2-regular digraph, $s_1$ can eventually learn the identities of the adversarial nodes immediately before and after it on the graph by sending probe messages.
Now consider the following scenario: 
a message from an honest user propagates on the line, and passes $s_1$. 
At an honest node $v$ between $s_1$ and the next adversarial node $s_2$ (see Figure \ref{fig:spies}), the message transitions into the spreading phase at and starts diffusing over the main P2P graph $H$. % \red{$H$ not defined?}
We assume that the adversary can reliably infer the diffusion source $v$.
Since $s_2$ did not receive the message before the spreading phase began, and $v$ was the source of the spreading phase, the adversary learns that $v$ lies between $s_1$ and $s_2$.
In this way, the adversary learns the internal, honest nodes of $G$ at a rate proportional to the creation of new transactions; by learning this graph, the adversary's expected per-node precision grows to $p$.

\begin{figure}[h]
    \centering
  \includegraphics[width=.23\textwidth]{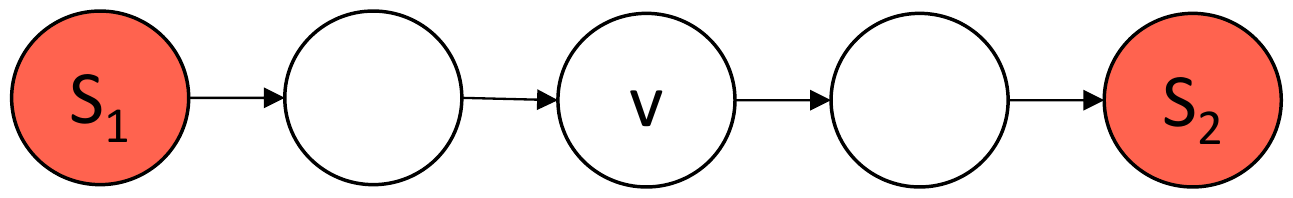}
  \caption{The adversary can easily learn line graphs.}
  \label{fig:spies}
\end{figure}

This problem must be managed by changing the graph quickly enough that the adversary cannot learn it---i.e., on the timescale at which transactions are executed.
As a ballpark estimate,
the Bitcoin network currently sees about three transactions per second \cite{blockchain_info}. 
Botnets can be rented through online services; for \$200, one can rent a botnet of 1,000 US-based \emph{zombies}, or corrupted hosts \cite{botnet_prices}.
Since the current Bitcoin network consists of about 5,500 servers \cite{bitnodes}, this corresponds to $p\approx 0.15$; 
each ward would have about seven nodes on average, of which five are unknown to the adversary in the fully-dynamic setting.
We conservatively assume that each transaction launches its spreading phase from a different honest node.
If we want to ensure that the adversary never learns more than 40\% of interior nodes, we should change the graph every $5500 \text{ transactions} \times \frac{5}{7} \times 0.4  \times \frac{1\text{ sec}}{3 \text{ transactions}}\approx 9$ minutes. 
This is easy to enforce in a distributed fashion; every nine minutes, each node will tear down its connections and form new ones. 
Synchrony between nodes is not needed for this restructuring due to the fully-distributed  line approximation protocol.

More powerful attackers can create botnets of tens of thousands of nodes, which would overwhelm \Algo. 
In such scenarios, statistical solutions are no longer appropriate. 

%% file: related.tex
\section{Related Work}
\label{sec:related}
Three key facets of this work appear in the literature: practical anonymity attacks on Bitcoin, source detection analysis in diffusion processes, and the design of anonymous communication protocols.

\noindent \textbf{Anonymity attacks on Bitcoin.}
There have been several attacks on the anonymity of Bitcoin, most of which harness the public nature of the blockchain \cite{reid2013analysis,ober2013structure,ron2013quantitative}.
Transaction patterns can be used to link user transactions over time, and in some cases identify the human owner of a public key.

More recently, authors have demonstrated deanonymization attacks on Bitcoin's networking stack. 
These attacks typically  use the first-spy estimator, and achieve surprisingly high accuracies \cite{koshy2014analysis,biryukov,biryukov2015bitcoin}.
The Bitcoin community has responded to these attacks with ad hoc changes to its networking stack for improved anonymity \cite{bitcoind}.
%We hope this work can serve as a principled guideline for how to change the networking stack.

\noindent \textbf{Analysis of diffusion.} A number of researchers have studied source detection on diffusion processes on graphs.
These results show that for various classes of graphs and adversarial models, reliable deanonymization is possible \cite{SZ11a,SZ12,FC12,WDZT14,LMOZ13,PVF12}.
However, there has been a relative lack of theoretical results in the analysis of diffusion under a spy-based adversary like ours.
Many of the results in this space propose effective heuristics that achieve high recall in practice \cite{PTV12,ZY13,ZCY14}.
These papers suggest that by using centrality information, adversaries may be able to launch stronger attacks than prior practical network attacks \cite{biryukov,biryukov2015bitcoin}.

\noindent \textbf{Anonymous broadcasting.} The topic of anonymous broadcasting has been studied for decades. 
The best-known example is dining cryptographer networks (DC nets), which enable a user to  broadcast a message anonymously with information-theoretic guarantees \cite{chaum88}.
DC nets are communication-intensive, which has prevented them from scaling beyond a few thousand nodes \cite{anonymousscale,corrigan2010dissent,golle2004dining}, and they are not well-suited to distributed systems like cryptocurrencies. 

Another relevant topic is adaptive diffusion (AD) \cite{KFSV14}, which was recently proposed as an anonymous spreading protocol for broadcasting content over fixed graphs. 
AD shares some properties with \Algo, such as symmetry-breaking.
However, AD was designed for social networks, which do not require all nodes to receive every message.  
Indeed, AD can `get stuck' on real graphs, meaning that some messages do not reach the entire network \cite{KFSV14}.
This property is unacceptable in cryptocurrencies: all nodes should receive all messages for fairness and consistency purposes. 

Finally, the core idea of \algo~spreading---i.e., passing content through a set of proxies---has been used in numerous anonymity systems, mainly for point-to-point communication \cite{tor,reiter1998crowds}. 
However, existing systems have not connected \algo~spreading to any fundamental anonymity guarantees, and they typically assume a complete graph topology \cite{tor}. 
In contrast, we identify topologies over which \algo~spreading actually provides strong guarantees (i.e., not complete graphs).
More fundamentally, our problem is focused on broadcasting over a network, which has different requirements and models than point-to-point messaging.
%Nonetheless, we have shown that complete topologies do not mix messages as well as other topologies.

%% file: conclusion.tex
\section{Conclusion}
\label{sec:conclusion}

In this paper, we redesign the Bitcoin P2P stack to provide anonymity against distributed adversaries (e.g., botnets) who wish to link users to their transactions. 
We present the \Algo~ networking policy, which achieves nearly-optimal anonymity guarantees with a simple, distributed implementation.
Intuitively, \Algo~ achieves these guarantees by mixing messages from different users on a graph that remains hidden from the adversary.
%We have also presented a framework for analyzing other networking policies in terms of precision and recall.
%This mixing makes it theoretically difficult for the adversary to jointly deanonymize users.

%A few obstacles to deployment remain. 
%First, we have assumed honest-but-curious adversarial nodes.
%In practice, botnet nodes may be running malicious code, thus allowing them to break protocol. 
%In this case, anonymity  can be negatively affected by nodes forwarding content inappropriately and/or misbehaving during the graph construction protocol.
%Hardening our protocols against such intrusions is critical.

A few topics of interest during  deployment are as follows: first, we have analyzed honest-but-curious adversarial nodes. In practice, botnet nodes may run malicious code that breaks protocol. 
In this case, anonymity can be negatively affected by nodes forwarding content inappropriately or misbehaving during graph construction. 
Hardening our protocols against such intrusions is critical.
For example, noninteractive graph construction protocols offer some robustness by reducing opportunities for the adversary to lie in order to generate an advantageous anonymity graph.
Alternatively, a system could use cryptographic proofs to ensure that nodes follow the graph construction protocol.
Exploring these options is beyond the scope of this paper. 

Second, we have paid less attention to the issue of message latency by assuming the two-phase architecture of dandelion spreading. As stated in Sec. \ref{sec:main}, it is unclear if such an architecture is inherently optimal. Understanding the tradeoff between anonymity and latency is of fundamental interest. 

%Second, we have ignored the issue of message latency by assuming the two-phase architecture of dandelion spreading.
%As mentioned in Sec. \ref{sec:main}, it is not  clear that such an architecture is inherently optimal. 
%Understanding the tradeoff between anonymity and latency is of fundamental interest.

%% file: appendix.tex
\section{Proofs}

\subsection{Section \ref{sec:bounds}: Anonymity Metric Properties }

\subsubsection{Proof of Theorem \ref{thm:bounds_prec_rec}}
\label{proof:bounds_prec_rec}
Consider any realization of the network, in which the messages $\mathcal{X}$ are mapped to the servers $V_H$ according to mapping rule $\mathtt{M}$. 
Then from the definition of precision and recall at any node $v$ (Equations \eqref{eq:precision}, \eqref{eq:recall}), we have 
\begin{align}
D_{\texttt{M}}(v) &= \frac{\mathbbm{1}\{\mathtt{M}(X_v) = v\}}{\sum_{w\in V_H} \mathbbm{1}\{ \mathtt{M}(X_w) = v\} }  \nonumber \\
& \leq \mathbbm{1}\{\mathtt{M}(X_v) = v\} = R_{\texttt{M}}(v). 
\end{align}
Hence it follows that the macro-averaged precision $D_\mathtt{M}$ is at most the macro-averaged recall $R_\mathtt{M}$, implying $\mathbf{D}_\mathtt{M} \leq \mathbf{R}_\mathtt{M}$.  

To prove inequality (b), let $V_\mathtt{M} = \{v\in V_H: \mathtt{M}(X_v) = v\}$ denote the set of servers whose corresponding messages are correctly mapped by $\mathtt{M}$.
Further, for each such node $v\in V_\mathtt{M}$, let $I_v = \{x\in\mathcal{X}:\mathtt{M}(x)=v, x\neq X_v\}$ denote all the messages (other than $v$'s own message $X_v$) that are mapped to $v$.   
Then, by definition we have $R_\mathtt{M} = |V_\mathtt{M}|/\tilde{n}$ and 
\begin{align}
\tilde{n}D_\mathtt{M} &= \sum_{v\in V_H}D_{\texttt{M}}(v) = \sum_{v\in V_\mathtt{M}} \frac{1}{|I_v|+1} \notag \\
&\geq \frac{|V_\mathtt{M}|^2 }{\sum_{v\in V_\mathtt{M}}(|I_v|+1)} \geq \frac{|V_\mathtt{M}|^2}{\tilde{n}} = \tilde{n}R_\mathtt{M}^2 \label{eq: prec recall eq}
\end{align}
where Equation~\eqref{eq: prec recall eq} follows from the arithmetic-mean harmonic-mean (A.M-H.M) inequality and $\sum_{v\in V_\mathtt{M}}(|I_v| + 1) \leq \tilde{n}$. 
Hence we have $R_\mathtt{M} \leq \sqrt{D_\mathtt{M}}$, which upon taking expectation and using Jensen's inequality, yields $\mathbf{R}_\mathtt{M} \leq \sqrt{\mathbf{D}_\mathtt{M}}$. %;  the theorem follows. 

\subsubsection{Proof of Theorem \ref{thm:lower_bounds_fs}}
\label{proof:lower_bounds_fs}

Recall that for honest server $v$, the tuple $(x,u,T_u(x))$ is contained in $S_v$ if $v$ forwards message $x$ to adversarial node $u$ at time $T_u(x)$. 
Let us now define a related quantity $\bar{S}_v$ to denote the set of messages $x\in\mathcal{X}$ forwarded by $v$ to some adversary such that $x$ was not received by any adversarial node previously. 
This quantity $\bar{S}_v$ is useful in analyses involving the first-spy estimator. 
$\bar{\mathbf{S}}$ denotes the vector of all $\bar{S}_v$'s. 
\begin{lemma} 
If $v\in V_H$ is a honest server node in a network with a fraction $p$ of adversaries, then the recall of the first-spy estimator is $\mathbf{R}_\mathtt{FS}(v) = \mathbb{P}(X_v \in \bar{S}_v) \geq p$.
\label{lem: lower bound}
\end{lemma}
\begin{proof}
Let $U \in \Gamma(v)$ denote the node to which $v$ first sends its message $X_v$. Then, 
\begin{align}
\mathbb{P}(U \in V_A) = \sum_{u\in V, u\neq v}\mathbb{P}(U = u)\mathbb{P}(U \in V_A | U = u) \notag \\ 
 = \sum_{u\in V, u\neq v}\frac{1}{n-1}\mathbb{P}(U\in V_A|U = u) = \frac{np}{n-1} \geq p,
 \end{align}
due to uniform distribution among the remaining nodes $V\backslash \{v\}$. Therefore we have, 
\begin{align} 
\mathbb{P}(X_v \in \bar{S}_v) \geq \mathbb{P}(U \in V_A) \geq p.
\end{align}
Thus $v$'s message is contained in $\bar{S}_v$ with probability at least $p$. The case where $v$ simultaneously broadcasts $X_v$ to multiple nodes can also be similarly bounded as above, and hence the lemma follows. 
\end{proof}

To show \eqref{eq:rec_lower}, note that $\mathbf{R}_\mathtt{OPT} \geq \mathbf{R}_\mathtt{FS}(v) \geq  p$, by Lemma \ref{lem: lower bound}.
Next, we show that the first-spy estimator also has a precision of at least $p^2$ regardless of the topology or  spreading scheme. 
Consider a random realization $\bar{\mathbf{S}}$, in which the adversaries observe a set of first-received messages $S_v\subseteq \mathcal{X}$ from each node $v\in V$.
Now, supposing in these observations there exists a subset of $t$ server nodes $\{v_1,v_2,\ldots,v_t\}$ whose own messages are included in the respective forwarded sets, i.e., $X_{v_i}\in \bar{S}_{v_i} \forall i=1,2,\ldots,t$. The macro-averaged precision in this case is 
\begin{align}
D_\mathtt{FS} = \frac{1}{\tilde n}\sum_{i=1}^t \frac{1}{|\bar{S}_{v_i}|} \geq \frac{t^2}{\tilde n \sum_{i=1}^{\tilde{n}} |\bar{S}_{v_i}|} \geq \frac{t^2}{\tilde{n}^2}, \label{eq: lower bound 1}
\end{align}
where the first inequality above is due to the arithmetic-mean harmonic-mean (A.M-H.M) inequality, and the second inequality is because the total number of messages is at most $\tilde n$. 
Equation~\eqref{eq: lower bound 1} in turn implies that  
\begin{align}
\mathbb{E}[D_\mathtt{FS}|T=t] \geq \frac{t^2}{\tilde{n}^2}.
\end{align}
The overall expected detection precision can then be bounded as 
\begin{align}
\mathbf{D}_\mathtt{FS} &= \mathbb{E}[D_\mathtt{FS}] = \sum_{t=0}^{\tilde{n}} \mathbb{P}(T=t)\mathbb{E}[D_\mathtt{FS}|T=t]  \notag \\
&\geq \sum_{t=0}^{\tilde{n}} \mathbb{P}(T=t) \frac{t^2}{\tilde{n}^2} = \frac{\mathbb{E}[T^2]}{\tilde{n}^2} \geq \frac{\mathbb{E}[T]^2}{\tilde{n}^2} \notag \\ 
&= \frac{\mathbb{E}[\sum_{v\in V_H}\mathbf{1}_{X_v\in \bar{S}_v}]^2}{\tilde{n}^2} \geq \frac{(p \tilde{n})^2}{\tilde{n}^2} = p^2 ,    \label{eq: lower bound 2}
\end{align}
where the inequality in Equation~\eqref{eq: lower bound 2} follows from Lemma~\ref{lem: lower bound}.  
Finally by definition we have $\mathbf{D}_\mathtt{OPT} \geq \mathbf{D}_\mathtt{FS}$ and hence the theorem follows.

\subsubsection{Proof of Theorem \ref{thm: optimal estimator}}
\label{proof: optimal estimator}

Let us first prove that the optimal mapping must be a matching. 
Supposing otherwise, consider a mapping $\mathtt{M}\in\mathcal{M}$ that is not a matching.
Then there exists a server $v$ that is mapped to the most number of messages $\{x_1,x_2,\ldots,x_k\}$ ($k>1$) in $\mathtt{M}$. 
This also implies there exists another node $u \in V_H$ such that no message is mapped to $u$. 
Now, the expected precision at server $v$ is given by 
\begin{align}
\mathbb{E}[D_\mathtt{M}(v)|\mathbf{O}] &= \frac{\sum_{i=1}^k \mathbb{P}(X_v=x_i|\mathbf{O})}{k} \notag \\
&\leq \max_{i\in\{1,\ldots,k\}} \mathbb{P}(X_v = x_i|\mathbf{O}).
\end{align}
On the other hand, the expected precision at $u$ is zero. 
Now, consider an alternative mapping $\mathtt{M'}\in\mathcal{M}$ in which all messages $x\in\mathcal{X}$ are mapped to servers exactly as in $\mathtt{M}$ except for the message $x_{i^*}$ where $i^* = \text{argmin}_{i\in\{1,\ldots,k\}}\mathbb{P}(X_v=x_i|\mathbf{O})$ which is mapped to server $u$. 
In this case, the expected precision at $v$ becomes
\begin{align}
\mathbb{E}[D_\mathtt{M'}(v)|\mathbf{O}] &= \frac{\sum_{i=1,i\neq i^*}^k \mathbb{P}(X_v=x_i|\mathbf{O})}{k-1} \notag \\
&\geq \mathbb{E}[D_\mathtt{M}(v)|\mathbf{O}],
\end{align}
while the expected precision at $u$ is 
\begin{align}
\mathbb{E}[D_\mathtt{M'}(u)|\mathbf{O}] = \mathbb{P}(X_u=x_{i^*}|\mathbf{O}) \geq 0. 
\end{align}
As such the total expected precision at servers $u$ and $v$ is 
\begin{align}
\mathbb{E}[D_\mathtt{M'}(v) + D_\mathtt{M'}(u) |\mathbf{O}] &\geq \mathbb{E}[D_\mathtt{M}(v) + D_\mathtt{M}(u) |\mathbf{O}] \notag \\
\Rightarrow \mathbb{E}[D_\mathtt{M'}|\mathbf{O}] &\geq \mathbb{E}[D_\mathtt{M}|\mathbf{O}].
\end{align}
Thus we have constructed a new mapping $\mathtt{M'}$ whose expected precision is at least as much as $\mathtt{M}$ and in which the maximum number of messages mapped to any server is smaller by 1.\footnote{In case of ties, we repeat the above process to each of the servers until the maximum server degree reduces by one.} Continuing this process, we conclude that for any mapping $\mathtt{M}\in\mathcal{M}$ there exists another matching mapping $\mathtt{M'}$ such that $\mathbb{E}[D_\mathtt{M'}|\mathbf{O}] \geq \mathbb{E}[D_\mathtt{M}|\mathbf{O}]$. Thus the optimizing mapping is achieved by a matching. 
 
Now, let $\mathcal{M^*}$ denote the set of all matchings in the bipartite graph $(V_H,\mathcal{X})$. 
By the first part of the theorem above, we can restrict our search to $\mathcal{M^*}$ for finding the optimal mapping. As such, 
\begin{align}
\mathbb{E}[D_\mathtt{OPT}|\mathbf{O}] &= \max_{\mathtt{M}\in\mathcal{M^*}} \mathbb{E}[D_\mathtt{M}|\mathbf{O}] \notag \\
&= \max_{\mathtt{M}\in\mathcal{M^*}} \sum_{(v,x)\in \mathtt{M}} \mathbb{P}(X_v = x|\mathbf{O}),
\end{align}  
implying that the optimum is achieved by a maximum weight matching.

\subsubsection{Proof of Corollary \ref{cor: opt upper bound}}
\label{proof: opt upper bound}
Let $\mathtt{M}\in\mathcal{M}$ be any mapping under observations $\mathbf{O}=(\mathbf{S},\mathbf{\Gamma})$. Consider a server $v$ and let $\{x_1,x_2,\ldots,x_k\}$ be the set of messages that are mapped to $v$ in $\mathtt{M}$. Then, 
\begin{align}
\mathbb{E}[D_\mathtt{M}(v)|\mathbf{O}] &\leq \frac{\sum_{i=1}^k \mathbb{P}(X_v=x_i|\mathbf{O})}{k}  \notag \\
&\leq \max_{i\in\{1,\ldots,k\}} \mathbb{P}(X_v = x_i|\mathbf{O}) \notag \\
&\leq \max_{x\in\mathcal{X}} \mathbb{P}(X_v = x|\mathbf{O}). \label{eq: cor upp bound}
\end{align}
Since the above Equation~\eqref{eq: cor upp bound} holds for any mapping $\mathtt{M}$,  it must hold for the optimal mapping as well.

\subsubsection{Proof of Theorem \ref{thm: optimal estimator recall}}
\label{proof: optimal estimator recall}

We want to prove that the optimal mapping must map each message $x\in\mathcal X$ to a server $v$ that maximizes $\prob(X_v=x|\mathbf{O})$. 
Supposing otherwise, let us consider a mapping $\mathtt{M}\in\mathcal{M}$ where there exists a server $w$ that is mapped to a set of messages $\{x_1,x_2,\ldots,x_k\}$ ($k\geq1$), 
where w.l.o.g. $w \notin \argmax_{v\in V_H} \prob(X_v=x_1|\mathbf{O})$. 
The expected recall at server $w$ is given by 
\begin{align}
\mathbb{E}[R_\mathtt{B}(w)|\mathbf{O}] &= \sum_{i=1}^k \mathbb{P}(X_w=x_i|\mathbf{O}). \notag 
\end{align}
Further, consider another node $u\in V_H$ such that \\ $u \in \argmax_{v\in V_H} \prob(X_v=x_1|\mathbf{O})$.
Suppose it is mapped to a different set of messages $\{y_1,\ldots, y_j\}$.
The expected recall for node $u$ is 
\begin{align}
\mathbb{E}[R_\mathtt{B}(u)|\mathbf{O}] &= \sum_{i=1}^j \mathbb{P}(X_u=y_i|\mathbf{O}). \notag 
\end{align}
Now, consider an alternative mapping $\mathtt{M'}\in\mathcal{M}$ in which all messages $x\in\mathcal{X}$ are mapped to servers exactly as in $\mathtt{M}$ except for the message $x_{1}$, which is mapped to server $u$. 
In this case, the expected recall at $w$ becomes
\begin{align}
\mathbb{E}[R_\mathtt{M'}(w)|\mathbf{O}] &= \sum_{i=2}^k \mathbb{P}(X_w=x_i|\mathbf{O}) \notag 
\end{align}
while the expected recall at $u$ is 
\begin{eqnarray}
\mathbb{E}[R_\mathtt{M'}(u)|\mathbf{O}] &=&  \mathbb{E}[R_\mathtt{M}(u)|\mathbf{O}]  + \mathbb{P}(X_u=x_{1}|\mathbf{O}).
\end{eqnarray}
As such the total expected precision at servers $u$ and $v$ is 
\begin{align}
\mathbb{E}[D_\mathtt{M'}(w) + D_\mathtt{M'}(u) |\mathbf{O}] &\geq \mathbb{E}[D_\mathtt{M}(w) + D_\mathtt{M}(u) |\mathbf{O}] \notag \\
\Rightarrow \mathbb{E}[D_\mathtt{M'}|\mathbf{O}] &\geq \mathbb{E}[D_\mathtt{M}|\mathbf{O}].
\end{align}
Thus we have constructed a new mapping $\mathtt{M'}$ whose expected precision is at least as much as $\mathtt{M}$ and in which the number of messages mapped to servers with sub-maximal likelihood is reduced by one. 
Continuing this process, we conclude that for any mapping $\mathtt{M}\in\mathcal{M}$ there exists another mapping $\mathtt{M'}$ such that each message $x$ is mapped to a server 
$v^*\in \argmax_{v\in V_H} \prob(X_v=x|\mathbf{O})$. 

\subsection{Section \ref{sec:algos}: Baseline Algorithms}

\subsubsection{Proof of Proposition \ref{thm:flooding static regular}}
\label{proof:flooding static regular}
Consider the broadcasting experiment on a random realization $G$ of the network topology. 
In the static case this topology is completely known to the adversary.  
As defined in Section~\ref{proof:lower_bounds_fs}, for any honest server $v\in V_H$, let $\bar{S}_v$ denote the set of transactions $x\in\mathcal{X}$ such that $x$ is directly forwarded by $v$ to some adversary and $x$ has not been received by any adversary previously. 
By our assumption on flooding, this means a server $v$'s message is contained in $\bar{S}_u$ if and only if 
\begin{itemize}
\item[(i)] $u$ is reachable from $v$, 
\item[(ii)] $u$ has an out-going edge to an adversary and 
\item[(iii)] no other node $u'\in V_H$ that satisfies the previous two conditions is strictly closer to $v$ than $u$. \end{itemize}
Thus by looking at the graph $G$, the adversary can construct a bipartite graph $B(V_H, V_H)$ in which there is edge $(u,v)\in V_H\times V_H$ if and only if $X_v$ will be contained in $\bar{S}_u$. 
Further for each $u\in V_H$, let $W_u = \{v\in V_H: (u,v)\in B, \nexists~ u'\neq u \text{ s.t. } (u',v)\in B \}$ denote the set of server nodes whose messages reach only the message set $\bar{S}_u$. 
Note that for servers $v\in V_H$ that have an out-going edge to an adversary, we must have $v\in W_v$. 

Now, once the messages have been broadcast in $G$, consider the following mapping strategy $\mathtt{M}$ for the adversary. 
First, for each set $\bar{S}_v$ we compute a subset $\bar{S}'_v = \{x\in \bar{S}_v: x \notin \bar{S}_{u}\forall u\neq v\}$.
Such a set $\bar{S}'_v$ corresponds to the messages that were delivered to the adversaries only by server $v$ and no other server. 
Thus, the messages in $\bar{S}'_v$ must precisely belong to the servers in $W_v$. 
As such, the adversary's mapping strategy can be: (i) for each $v\in V_H$ pick a random matching mapping between $\bar{S}'_v$ and $W_v$ and (ii) assign any remaining messages randomly to the remaining server nodes. 
Let $\mathcal{E}_v$ denote the event that $v$ has an out-going edge to an adversary. 
The payoff can then be bounded as
\begin{align}
\mathbb{E}[\tilde{n}D_\mathtt{M}|G] &\geq \mathbb{E}\left[ \sum_{\substack{v\in V_H : \\ |W_v|\geq 1}}\sum_{u\in W_v}\mathbbm{1}\{\mathtt{M}(X_u) = u\} \bigg| G \right] \notag \\
&= \sum_{\substack{v\in V_H : \\ |W_v|\geq 1}}\sum_{u\in W_v} \mathbb{E}[\mathbbm{1}\{\mathtt{M}(X_u) = u\} | G] \notag \\
&= \sum_{\substack{v\in V_H : \\ |W_v|\geq 1}}\sum_{u\in W_v} \frac{1}{|W_v|} = \sum_{v\in V_H}\mathbbm{1}\{ |W_v| \geq 1\}, \label{eq: static flood}
\end{align}
where Equation~\eqref{eq: static flood} follows because in a random matching any message in $S'_v$ is likely to be assigned to its true server in $W_v$ with probability $1/|W_v|$. Hence the total average precision is bounded by
\begin{align}
\tilde{n}\mathbf{D}_\mathtt{M} &\geq \mathbb{E}\left[\sum_{v\in V_H}\mathbbm{1}\{ |W_v| \geq 1\}\right] \\
&\geq \sum_{v\in V_H} \mathbb{P}(\mathcal{E}_v) 
= \tilde{n}(1- (1-p)^d)
\end{align}
and we have the proposition.

\begin{figure}[t]
    \centering
  \includegraphics[width=.21\textwidth]{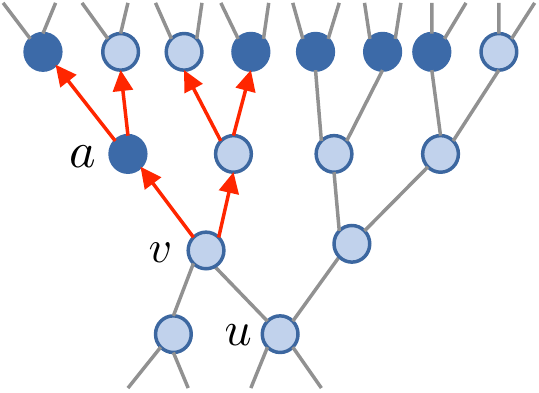}
  ~
  \includegraphics[width=.21\textwidth]{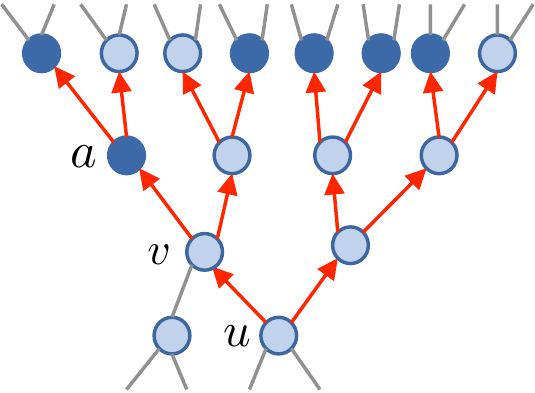}
  \caption{Comparison of the number of adversarial nodes receiving a message in $2$ rounds following first reception when the source is (i) directly connected to an adversary (left) and (ii) away from the adversary (right). The propagation of the message is shown in red; darkened nodes are adversarial. 
 }
  \label{fig:flooding_propagation}
\end{figure}

\subsubsection{Proof Sketch of Proposition \ref{thm:flooding dynamic regular}}
\label{proof:flooding dynamic regular}
Before we begin the proof, first notice that in the dynamic setting the adversaries are directly connected to at most $dpn$ (i.e., roughly a fraction $p$) honest servers, while the rest of the server locations are unknown to the adversary.
Since the hidden servers can only be trivially de-anonymized, in order to obtain our claimed average precision of $O(p)$ each of the servers visible to the adversary must be de-anonymized with a high precision close to 1. 
Indeed, in the following we describe a simple mapping scheme $\mathtt{M}$ that achieves this high precision. 

For a server $v\in V_H$, let $\mathcal{E}_v$ denote the event that at least one of $v$'s out-going edges is connected to an adversary. 
Consider then, the spreading of $v$'s message $X_v$ in the graph under event $\mathcal{E}_v$.
Since $G$ is a random $d$-regular graph, using the result in~\cite{makover2006regular}, there is almost surely a regular tree of depth at least $\frac{1}{2} \log_{d-1} n$ rooted at $v$.
For simplicity, let us consider $d=4$ in which case there is a tree of depth at least $\frac{1}{4}\log n$ rooted at $v$ almost surely.   
Thus $v$'s message propagates along this tree, reaching two nodes (at least one of which is an adversary, due to $\mathcal{E}_v$) in the first round and subsequently reaching $2^i$ new nodes in the $i$-th round for each $i<\frac{1}{4}\log n$. 
Since a fraction $p$ of the nodes are adversarial, this implies in the $i$-th round we expect roughly $p2^i$ adversarial nodes to receive $X_v$. 
On the other hand, if some other server $u$ upstream of $v$ had started broadcasting its message $X_u$, then more adversarial nodes would have received $X_u$ in the $i$-th round following reception by the first adversarial node (see Figure~\ref{fig:flooding_propagation}). 

The above observation then, naturally motivates a mapping $\mathtt{M}$ as described in Algorithm~\ref{algo:dreg_dyn_map}. 
\begin{algorithm}[t]
\DontPrintSemicolon
\KwIn{Time-stamp $T_v(x)$ and sender $S_v(x)$ for each message $x\in\mathcal{X}$ received by adversary $v$ for all $v\in V_A$.}
\KwOut{Mapping $\mathtt{M}$ from $\mathcal{X}$ to $V_H$}
$I\leftarrow V_H, J\leftarrow \mathcal{X}$ \; 
\For{each $x\in \mathcal{X}$} {
$a_\mathrm{init} \leftarrow \argmin_{v\in V_A} T_v(x)$ \;
$T_\mathrm{init} \leftarrow T_{a_\mathrm{init}}(x) $ \;
$v_\mathrm{init} \leftarrow S_{a_\mathrm{init}}(x)$\;
$\eta \leftarrow |\{v\in V_A: T_v(x) = T_\mathrm{init} + \frac{1}{4}\log n - 1\}|$ \;
\If{$\eta < 2pn^{1/4} $}{
	$\mathtt{M}(x)\leftarrow v_\mathrm{init}$\;
	$I \leftarrow I \backslash \{v_\mathrm{init}\}$ \;
	$J \leftarrow J \backslash \{x\}$ \;
	}
}
Randomly assign messages in $J$ to servers in $I$\;
\Return{$\mathtt{M}$}\;
\caption{Mapping algorithm under flooding for a dynamic $4$-regular graph.}
\label{algo:dreg_dyn_map}
\end{algorithm}
In this strategy, the adversary simply counts the number of adversarial nodes that received a particular message $x\in\mathcal{X}$ at a time $\frac{1}{4}\log n - 1$ rounds after  the message was first received by some adversarial node $a$.
If this number is small ($<2pn^{1/4}$) then we conclude the source of $x$ to be the server $v$ that sent the message to the first adversary. Otherwise the message is randomly assigned to an unassigned server at the end. 

The algorithm works because if $v$ were truly the source of $x$, then in the $(\frac{1}{4}\log n - 1)$-th round following reception by $a$, the number of adversarial nodes to receive $x$ is less than $2pn^{1/4}$ with a probability at least $1-2^{-\log(4/e) p n^{1/4}}$ by the Chernoff bound. On the other hand if $v$ were not the true source of $x$, then $x$ was initially broadcast at a time at least 2 rounds before $a$ received it. 
This implies at least $2pn^{1/4}$ adversarial nodes receive the message at a time $\frac{1}{4}\log n -1$ rounds following reception by $a$. 
Thus the total probability of error can be bounded by the union bound, to yield that whenever a server $v$'s out-going edges are connected to at least one adversarial node, $X_v$ is mapped to $v$ with precision 1 with high probability. 
Such an event $\mathcal{E}_v$ happens with a probability at least $p$ to conclude the Theorem. 

\subsubsection{Proof of Proposition \ref{thm:complete}}
\label{proof:complete}
For any message $X_u$, let $\Pi_u = (\Pi_{1,u},\Pi_{2,u},\ldots,\Pi_{L_u,u})$ be the path taken by a message from its source $u$ $(=\Pi_{1,u})$ until it reaches an adversarial node $\Pi_{L_u,u}$ for the first time ($L_u$ denotes the length of the path). 
Further for any two nodes $v,u\in V_H$, let $\mathcal{E}_{u,v}$ denote the event that $u$'s message $X_u$ reaches the adversary through server $v$, i.e., $\Pi_{1,u} = u, \Pi_{2,u} \notin V_A, \Pi_{3,u}\notin V_A, \ldots, \Pi_{k-2,u}\notin V_A, \Pi_{k-1, u} = v$ and $\Pi_{k,u}\in V_A$.  
Then by counting over paths of all possible lengths, we can evaluate probability of $\mathcal{E}_{u,u}$ as 
\begin{align}
\mathbb{P}(\mathcal{E}_{u,u}) &= \sum_{l \geq 2} \mathbb{P}(L_k = l, \mathcal{E}_{u,u}) \notag \\
&= \left( \frac{np}{n} \right) +  \sum_{l \geq 3} \left( \frac{\tilde{n}}{n} \right)^{l-3}\left( \frac{1}{n} \right) \left( \frac{np}{n} \right) = p + \frac{1}{n}. 
\end{align}
Similarly, for $u\in V_H, u\neq v$, 
\begin{align}
\mathbb{P}(\mathcal{E}_{u,v}) = \sum_{l\geq 3}\left( \frac{\tilde{n}}{n} \right)^{l-3} \left( \frac{1}{n} \right) \left( \frac{np}{n} \right) = \frac{1}{n}. \label{eq: comp graph prob}
\end{align}
Further, since the messages are all forwarded independently the set of events $\{\mathcal{E}_{v,u}: v\in V_H\}$  are mutually independent for each server $u\in V_H$. 
Hence the expected cost incurred at a server under the first-spy estimator can be written as 
\begin{align}
\mathbf{D}_\mathtt{FS}(v) &= \left( p + \frac{1}{n} \right) \mathbb{E}\left[\frac{1}{1+Z_v} \right] \notag \\
&= \left( p + \frac{1}{n} \right) \frac{1}{\tilde{n}\frac{1}{n}} \left(1-\left(1-\frac{1}{n}\right)^{\tilde{n}}\right),
\end{align}
where $Z_v = \sum_{u\in V_H, u\neq v} \mathbbm{1}\{\mathcal{E}_{u,v}\}$ is the number of messages, other than $X_v$, that reach the adversary through $v$ and $Z_v \sim \text{Binom}\left(\tilde{n}-1,\frac{1}{n}\right)$ because of independence of messages and Equation~\eqref{eq: comp graph prob}. 
The last equation above can be further simplified to yield the bound
\begin{align}
\mathbf{D}_\mathtt{FS}(v) \geq \frac{p}{1-p}(1-e^{p-1}),
\end{align}
which when averaged over all honest nodes $v\in V_H$ gives us the desired result. 

\subsection{Section \ref{sec:main}: Main Result - Dandelion}

\subsubsection{Proof of Theorem \ref{thm:dandelion recall}}
\label{proof:dandelion recall}
We first show that the first-spy estimator is recall-optimal for \algo~spreading, then that the first-spy estimator has an expected recall of $p$.

To show the first step, i.e., $\mathbf R_{\texttt{OPT}}=\mathbf R_{\texttt{FS}}$, 
Theorem \ref{thm: optimal estimator recall} implies that we must show that for every message $x$, its exit node $z$ (i.e., the node implicated by the first-spy estimator) maximizes $\prob(X_v = x|\mathbf{O})$.
For any message $X_u$, let $\Pi_u = (\Pi_{1,u},\Pi_{2,u},\ldots,\Pi_{L_u,u})$ be the path taken by a message from its source $u$ $(=\Pi_{1,u})$ until it reaches an adversarial node $\Pi_{L_u,u}$ for the first time ($L_u$ denotes the length of the path). 
From the adversary's observation $\mathbf S$, $\Pi_{L_u-1,u}$ and $\Pi_{L_u,u}$ are fixed as the exit node $z$ and the first spy for $X_u$, respectively.
Due to the specification of \algo~spreading (Algorithm \ref{algo:dandelion}), the likelihood of this path, $\mathcal L(\Pi_u)$, is 
$
\mathcal L(\Pi_u) = \prod_{i = 1}^{L_u-1} \frac{1}{\text{deg}(\Pi_{i,u})},
$
where deg$(v)$ denotes the out-degree of $v$.
Assuming a uniform prior over candidate sources, we have $\prob(X_v = x|\mathbf{O}) \propto \mathcal L(\Pi_v)$. Since each node is assumed to have an out-degree of at least 1, this likelihood is maximized by taking the shortest path possible. That is,  the maximum-likelihood path over all paths originating at honest candidate sources gives $z \in \argmax_{v\in V_H} \prob(X_v = x|\mathbf O)$.
Hence the first-spy estimator is also a maximum-recall estimator. 

Now we analyze the recall of the first-spy estimator.
Let $\mathcal{P}_v$ denote the event that $v$'s parent (i.e., the next node in the line) is adversarial. 
Then the expected recall is
\begin{align*}
\mathbf R_{\texttt{OPT}} = \mathbb{E}[R_\mathtt{FS}|\mathbf{S},G] &= \frac{1}{\tilde n}\E \left [\sum_{v\in V_H} \mathbbm{1}\{\mathcal{P}_v\} \right ] \\
\Rightarrow \mathbf{R}_\mathtt{OPT} = \frac{1}{\tilde n}\sum_{v\in V_H} \mathbb{P}(\mathcal{P}_v) &=\frac{1}{\tilde n}\sum_{v\in V_H} \left( \frac{n-1}{n} \frac{np}{n-1} \right) = p.
\end{align*}

\subsubsection{Proof of Proposition \ref{thm:static tree}}
\label{proof:static tree}

For any node $v$, let $W_v = \{u\in V_H: X_u \in S_v\}$ denote the ward under node $v$ and let $\mathbf{W}$ denote the set of all wards. 
Note that in the tree topology, the wards $W_v$ can be completely determined from knowledge of the graph $G$.  
Also, let $I(v)$ denote the node $u\in V_H$ such that $X_v \in S_u$. Then, the expected cost at a node $v$ under observations $\mathbf{O}=(\mathbf{S},G)$ by the adversary can be written as
\begin{align}
\mathbb{E}[D_\mathtt{M}(v)|\mathbf{S},G] = \mathbb{E}[D_\mathtt{M}(v)|\mathbf{S},G,\mathbf{W}] = \frac{1}{|W_{I(v)}|}. \label{eq: tree ward prob} 
\end{align}
This follows because the matching estimator $\mathtt{MAT}$ assigns the messages in $S_v$ to the nodes in $W_v$ as a random matching, and hence the probability of a node receiving the correct message is $1/|W_v|$.  
Summing Equation~\eqref{eq: tree ward prob} over all honest nodes, we have 
\begin{align}
\mathbb{E}[D_\mathtt{MAT}|\mathbf{S},G] = \sum_{v\in V_H} \frac{1}{|W_{I(v)}|} = |W|,  \label{eq: tree ward number}
\end{align}
where $|W| = |\{v:W_v \neq \emptyset \}|$ denotes the number of non-empty wards, and $\emptyset$ denotes the null set. 
Now, let $\mathcal{P}_v$ denote the event that $v$'s parent is adversarial. 
Since a ward under a node $v$ is non-empty iff $v$'s parent is adversarial, Equation~\eqref{eq: tree ward number} above becomes
\begin{align}
\mathbb{E}[D_\mathtt{MAT}|\mathbf{S},G] &= \frac{1}{\tilde n}\E \left [\sum_{v\in V_H} \mathbbm{1}\{\mathcal{P}_v\} \right ] \\
\Rightarrow \mathbf{D}_\mathtt{MAT} = \frac{1}{\tilde n}\sum_{v\in V_H} \mathbb{P}(\mathcal{P}_v) &=\frac{1}{\tilde n}\sum_{v\in V_H} \left( \frac{n-1}{n} \frac{np}{n-1} \right) = p.   \notag 
\end{align}

\subsubsection{Proof of Proposition \ref{thm: dynamic trees}}
\label{proof: dynamic trees}
First note that since the tree is dynamic, the adversary's observations consists of $\mathbf{O}=(\mathbf{S},\Gamma(V_A))$, i.e., the transaction logs and the local neighborhood of adversarial nodes. 
Now for any honest node $v\in V_H$, let $\mathcal{E}_v$ denote the event that (i) $v$ occurs at a position in $G$ which is a leaf of the tree and (ii) $v$'s parent is an adversary. 
Similarly let $\mathcal{I}_v$ denote the event that $v\in V_H$ occurs at the interior of the tree. 
We first show that whenever $\mathcal{E}_v$ happens, $v$ is detected with certainty under the first-spy estimator, i.e., 
\begin{align}
\mathbb{E}[D_\mathtt{FS}(v)|\mathbf{S},\Gamma(V_A),\mathcal{E}_v] = 1. 
\end{align} 
This is because $D_\mathtt{FS}(v) = \frac{\sum_{x\in S_v}\mathbbm{1}\{X_v = x\}}{|S_v|}$ in the first-spy estimator and $S_v = \{X_v\}$ whenever $\mathcal{E}_v$ happens. As such, 
\begin{align}
\mathbb{E}[D_\mathtt{FS}(v)|\mathbf{S},\Gamma(V_A),\mathcal{E}_v] &= \mathbb{E}[\mathbbm{1}\{X_v = X_v\}|\mathbf{S},\Gamma(V_A),\mathcal{E}_v] = 1 \notag  \\
\Rightarrow \mathbb{E}[D_\mathtt{FS}(v)|\mathcal{E}_v] &= 1. 
\end{align}
Hence the expected payoff becomes
\begin{align}
\mathbf{D}_\mathtt{FS}(v) &= \mathbb{P}(\mathcal{E}_v)\mathbb{E}[D_\mathtt{FS}(v)|\mathcal{E}_v] + \mathbb{P}(\mathcal{I}_v)\mathbb{E}[D_\mathtt{FS}(v)|\mathcal{I}_v] \\
&\geq \mathbb{P}(\mathcal{E}_v)\mathbb{E}[D_\mathtt{FS}(v)|\mathcal{E}_v] = \frac{1}{2}\frac{np}{(n-1)} \geq \frac{p}{2},
\end{align}
since at least half of the nodes are leaves in a perfect % balanced 
$d$-ary tree.
Summing over all honest nodes gives the result.

\subsubsection{Proof of Theorem \ref{thm: dynamic line}}
\label{proof: dynamic line}

As in the case of dynamic trees, the adversary's observations consists of $\mathbf{O}=(\mathbf{S},\Gamma(V_A))$ in the dynamic line as well. 
The proof works by evaluating the cost incurred under various possibilities for the local neighborhood structure around a node in the network. 
For any honest server node $v\in V_H$, let $\mathcal{E}_v(i,j)$ denote the event that (i) $i$ nodes preceding $v$ are honest nodes, the $(i+1)$-th node preceding $v$ is adversarial and (ii) $j$ nodes succeeding $v$ are honest nodes and the $(j+1)$-th node following $v$ is adversarial. 
Also for ease of notation let $\mathcal{I}_v$ denote the event $\cup_{i>0,j>0}\mathcal{E}_v(i,j)$. 
Then the following lemmas hold true. 
\begin{lemma} \label{lem: line graph expect val}
On a line-graph, for any $i,j>0$, we have 
\begin{align}
&\mathbb{E}[\max_{x\in\mathcal{X}}\mathbb{P}(X_v=x|\mathbf{S}, \Gamma(V_A),\mathcal{E}_v(i,0))|\mathcal{E}_v(i,0)]  \leq \frac{1}{i+1} \notag \\
&\mathbb{E}[\max_{x\in\mathcal{X}}\mathbb{P}(X_v=x|\mathbf{S}, \Gamma(V_A),\mathcal{E}_v(0,j))|\mathcal{E}_v(0,j)]  \leq \frac{1}{j+1} \notag \\ 
 &\mathbb{E}[\max_{x\in\mathcal{X}}\mathbb{P}(X_v=x|\mathbf{S}, \Gamma(V_A),\mathcal{E}_v(0,0))|\mathcal{E}_v(0,0)]  \leq 1 \notag \\
&\mathbb{E}[\max_{x\in\mathcal{X}}\mathbb{P}(X_v=x|\mathbf{S}, \Gamma(V_A),\mathcal{I}_v)|\mathcal{I}_v] \leq \frac{1}{n(1-3p)}. 
\end{align}
\end{lemma}
\begin{proof}
Consider a realization $G$ of the network topology such that our desired event $\mathcal{E}_v(i,0)$ happens. 
In such a graph $G$, the node succeeding $v$ is an adversarial node and the $i$ nodes preceding $v$ are honest. Let us denote this set of $i+1$ nodes -- comprising of the $i$ nodes preceding $v$ and $v$ itself -- as $W_v$ (i.e., the ward of $v$). 
Now, if the messages assigned to the nodes outside of $W_v$ is denoted by $X(V_H \backslash W_v)$, then for any $x\in S_v$ we have $\mathbb{P}(X_v = x|G,\mathbf{S},\Gamma(V_A),\mathcal{E}_v(i,0),X(V_H\backslash W_v)$
\begin{align}
 &= \frac{\mathbb{P}(X_v=x,\mathbf{S},X(V_H\backslash W_v)|G,\Gamma(V_A),\mathcal{E}_v(i,0))}{\sum_{x\in S_v}\mathbb{P}(X_v=x,\mathbf{S},X(V_H\backslash W_v)|G,\Gamma(V_A),\mathcal{E}_v(i,0))}  \notag \\
&= \frac{\mathbb{P}(X_v=x,X(V_H\backslash W_v)|G,\Gamma(V_A),\mathcal{E}_v(i,0))}{\sum_{x\in S_v}\mathbb{P}(X_v=x,X(V_H\backslash W_v)|G,\Gamma(V_A),\mathcal{E}_v(i,0))} \notag \\
&= \frac{1}{i+1}, \label{eq: line proof val 1}
\end{align}
by using the fact that the allocation of messages $\mathbf{X}$ is independent of the graph structure $(G,\Gamma(V_A),\mathcal{E}_v(i,0))$ and $\mathbb{P}(\mathbf{S}|X_v=x, X(V_H\backslash W_v), G, \Gamma(V_A), \mathcal{E}_v(i,0))=1$ on a line-graph. 
Now, taking expectation on both sides of Equation~\eqref{eq: line proof val 1} we get
\begin{align}
\mathbb{P}(X_v = x|\mathbf{S},\Gamma(V_A),\mathcal{E}_v(i,0)) = \frac{1}{i+1}~\forall x \in S_v \notag \\
\Rightarrow \max_{x\in\mathcal{X}} \mathbb{P}(X_v = x|\mathbf{S},\Gamma(V_A),\mathcal{E}_v(i,0)) = \frac{1}{i+1} \text{ or } \notag \\
 \mathbb{E}[\max_{x\in\mathcal{X}}\mathbb{P}(X_v=x|\mathbf{S}, \Gamma(V_A),\mathcal{E}_v(i,0))|\mathcal{E}_v(i,0)]  = \frac{1}{i+1}.
\end{align}
By a similar argument as above, we can also show that 
\begin{align}
 \mathbb{E}[\max_{x\in\mathcal{X}}\mathbb{P}(X_v=x|\mathbf{S}, \Gamma(V_A),\mathcal{E}_v(0,j))|\mathcal{E}_v(0,j)]  &= \frac{1}{j+1}, \notag \\ 
 \mathbb{E}[\max_{x\in\mathcal{X}}\mathbb{P}(X_v=x|\mathbf{S}, \Gamma(V_A),\mathcal{E}_v(0,0))|\mathcal{E}_v(0,0)]  &= 1.
\end{align}
Finally let us consider the case where $v$ is an interior node, i.e., event $\mathcal{I}_v$ happens. 
As before, for a head-node $u$ (an honest node whose successor is an adversarial node) let $W_u$ denote the ward containing $u$.   
Notice that under observations $\mathbf{S},\Gamma(V_A)$ the adversaries know (i) the head and tail nodes of each ward (from $\Gamma(V_A)$) and (ii) the size of each ward ($|W_u| = |S_u|$).
Therefore if a message $x$ is such that $x\in S_u$ for some $u$, then 
\begin{align}
\mathbb{P}&(X_v = x|\mathbf{S},\Gamma(V_A),\mathcal{I}_v) = \mathbb{P}(X_v = x, v\in W_u|\mathbf{S},\Gamma(V_A),\mathcal{I}_v) \notag \\
&= \mathbb{P}(v\in W_u|\mathbf{S},\Gamma(V_A),\mathcal{I}_v) \mathbb{P}(X_v = x| v\in W_u,\mathbf{S},\Gamma(V_A),\mathcal{I}_v) \notag \\
&= \frac{|W_u|-2}{|I|} \frac{1}{|W_u|} \leq \frac{1}{|I|} \leq \frac{1}{n(1-3p)},
\end{align} 
where $I$ denotes the set of all interior nodes and $|I| \geq n(1-3p)$ since each adversary is a neighbor to at most 2 honest server nodes. 
Hence we have 
\begin{align}
\mathbb{E}[\max_{x\in\mathcal{X}}\mathbb{P}(X_v=x|\mathbf{S}, \Gamma(V_A),\mathcal{I}_v)|\mathcal{I}_v] \leq \frac{1}{n(1-3p)},
\end{align}
concluding the proof. 
\end{proof}

\begin{lemma} \label{lem: line graph prob val}
On a line-graph, for $i,j>0$ we have
\begin{align}
\mathbb{P}(\mathcal{E}_v(i,0)) &\leq \left(p+\frac{1}{n}\right)^2\left(1-p+\frac{2}{n}\right)^i \label{eq: lem event 1}\\
\mathbb{P}(\mathcal{E}_v(0,j)) &\leq \left(p+\frac{1}{n}\right)^2\left(1-p+\frac{2}{n}\right)^j \label{eq: lem event 2} \\
\mathbb{P}(\mathcal{E}_v(0,0)) &\leq (p+1/n)^2 \label{eq: lem event 3}\\
\mathbb{P}(\mathcal{I}_v) &\leq (1-p)^2.  \label{eq: lem event 4}
\end{align}
\end{lemma}
\begin{proof}
First let us consider the event $\mathcal{E}_v(i,0)$ in which node $v$ has an adversarial successor, $i$ honest predecessor nodes and an adversarial $i+1$-th predecessor.
Let $Y_v$ denote the position of node $v$ in the line graph. Then
\begin{align}
\mathbb{P}(\mathcal{E}_v(i,0)) = \sum_{j=i+1}^{n}\mathbb{P}(Y_v=j)\mathbb{P}(\mathcal{E}_v(i,0)|Y_v=j),  \label{eq: counting eq}
\end{align}
since $v$ needs to be at a position on the line graph where at least $i+1$ predecessors are feasible. 
Now, for $i+1\leq j \leq n$, by a simple counting argument we have $\mathbb{P}(\mathcal{E}_v(i,0)|Y_v=j) =$ 
\begin{align}
\left(\frac{np}{n-1}\right)\left(\frac{np-1}{n-2}\right) \left(\frac{\tilde{n}-1}{n-3}\right)\left(\frac{\tilde{n}-2}{n-4}\right)\ldots \left(\frac{\tilde{n}-i}{n-i-2}\right) \notag \\
\leq \left(p+\frac{1}{n}\right)^2\left(1-p+\frac{2}{n}\right)^i \notag
\end{align}
Combining the above inequality with Equation~\eqref{eq: counting eq} we conclude that
\begin{align}
\mathbb{P}(\mathcal{E}_v(i,0)) \leq \left(p+\frac{1}{n}\right)^2\left(1-p+\frac{2}{n}\right)^i
\end{align}
for $i>0$. 
By essentially a similar counting as above, we can also obtain the remaining Equations~\eqref{eq: lem event 2},~\eqref{eq: lem event 3} and~\eqref{eq: lem event 4} from the Lemma.   
\end{proof}

\begin{lemma} \label{lem: conditioning events}
If $\mathcal{E}_1,\mathcal{E}_2,\ldots,\mathcal{E}_k$ is a set of mutually exclusive and exhaustive events, and $v\in V_H$ is any honest server node, then
\begin{align}
\mathbf{D}_\mathtt{OPT}(v) \leq \sum_{i=1}^k \mathbb{P}(\mathcal{E}_i) \mathbb{E}[\max_{x\in\mathcal{X}}\mathbb{P}(X_v=x|\mathbf{S}, \Gamma(V_A),\mathcal{E}_i)|\mathcal{E}_i].
\end{align}
\end{lemma}
\begin{proof}
The proof is straightforward and follows from Corollary~\ref{cor: opt upper bound}. 
From Equation~\eqref{eq: opt upper bound} we have 
\begin{align}
\mathbb{E}&[D_\mathtt{OPT}(v)|\mathbf{S},\Gamma(V_A)] \leq \max_{x\in\mathcal{X}}\mathbb{P}(X_v=x|\mathbf{S},\Gamma(V_A)) \notag \\
&= \max_{x\in\mathcal{X}} \sum_{i=1}^k \mathbb{P}(\mathcal{E}_i|\mathbf{S},\Gamma(V_A))\mathbb{P}(X_v=x|\mathbf{S},\Gamma(V_A),\mathcal{E}_i) \notag \\
&\leq \sum_{i=1}^k \mathbb{P}(\mathcal{E}_i|\mathbf{S},\Gamma(V_A)) \max_{x\in\mathcal{X}} \mathbb{P}(X_v=x|\mathbf{S},\Gamma(V_A),\mathcal{E}_i).
\end{align}
Taking expectation on both sides of the above equation, we get $\mathbf{D}_\mathtt{OPT}(v) \leq \mathbb{E}[\max_{x\in\mathcal{X}} \mathbb{P}(X_v=x|\mathbf{S},\Gamma(V_A),\mathcal{E}_i)]$,
\begin{align}
\Rightarrow \mathbf{D}_\mathtt{OPT}(v) \leq  \sum_{i=1}^k \mathbb{P}(\mathcal{E}_i) \mathbb{E}[\max_{x\in\mathcal{X}}\mathbb{P}(X_v=x|\mathbf{S}, \Gamma(V_A),\mathcal{E}_i)|\mathcal{E}_i], \notag
\end{align}
and thus proving the Lemma. 
\end{proof}

%\begin{proof}

To complete the proof of the Theorem, let use Lemma~\ref{lem: conditioning events} with $\mathcal{E}_v(i,0),\mathcal{E}_v(0,j),\mathcal{E}_v(0,0)$ and $\mathcal{E}_v$ for $i,j>0$ as the set of mutually exclusive and exhaustive events. 
Then the expected payoff at $v$ can be bounded as $\mathbf{D}_\mathtt{OPT}(v) \leq $  
\begin{align}
\sum_{i > 0} \mathbb{P}(\mathcal{E}_v(i,0)) \mathbb{E}[\max_{x\in\mathcal{X}}\mathbb{P}(X_v=x|\mathbf{S}, \Gamma(V_A),\mathcal{E}_v(i,0))|\mathcal{E}_v(i,0)] \notag \\
+ \sum_{j > 0} \mathbb{P}(\mathcal{E}_v(0,j)) \mathbb{E}[\max_{x\in\mathcal{X}}\mathbb{P}(X_v=x|\mathbf{S}, \Gamma(V_A),\mathcal{E}_v(0,j))|\mathcal{E}_v(0,j)] \notag \\
+ \mathbb{P}(\mathcal{E}_v(0,0)) \mathbb{E}[\max_{x\in\mathcal{X}}\mathbb{P}(X_v=x|\mathbf{S}, \Gamma(V_A),\mathcal{E}_v(0,0))|\mathcal{E}_v(0,0)] \notag \\
+ \mathbb{P}(\mathcal{I}_v) \mathbb{E}[\max_{x\in\mathcal{X}}\mathbb{P}(X_v=x|\mathbf{S}, \Gamma(V_A),\mathcal{I}_v)|\mathcal{I}_v], \label{eq: line graph main thm}
\end{align}
where the values of the individual expectation and probability terms in the above Equation~\eqref{eq: line graph main thm} have been computed in Lemmas~\ref{lem: line graph expect val} and~\ref{lem: line graph prob val} respectively. 
Using those bounds, we get 
\begin{align}
\mathbf{D}_\mathtt{OPT}(v) \leq &  \sum_{i>0} \left(p+\frac{1}{n}\right)^2\left(1-p+\frac{2}{n}\right)^i \frac{1}{i+1} \notag \\
&+ \sum_{j>0} \left(p+\frac{1}{n}\right)^2\left(1-p+\frac{2}{n}\right)^j \frac{1}{i+1} \notag \\
&+ (p+1/n)^2 + (1-p)^2 \frac{1}{n(1-3p)} \notag \\
\leq & \frac{2(p+\frac{1}{n})^2}{(1-p+\frac{2}{n})} \log \left( \frac{1}{p-\frac{2}{n}} \right) + \frac{(1-p)^2}{n(1-3p)} \notag \\
\leq & \frac{2p^2}{1-p}\log\left( \frac{2}{p} \right) + O\left(\frac{1}{n}\right).
\end{align} 
Finally averaging the expected payoff $\mathbf{D}_\mathtt{OPT}(v)$ over each of the $\tilde{n}$ honest server nodes $v\in V_H$, we get the desired result. 
%\end{proof}

\subsection{Section \ref{sec:systems}: Systems Issues}

\subsubsection{Proof of Proposition \ref{prop:twochoices}}
\label{proof:twochoices}
We map the problem of constructing a line (i.e., a $2$-regular digraph) to one of assigning balls to bins. 
Suppose each ball represents an outgoing connection, and each bin represents a server who may accept that outgoing connection.
There are $n$  balls and $n$ bins; for the sake of simplicity, we assume that a server can establish a connection to itself, so all $n$ bins are available to each ball.
Then the maximum degree $d_m$ of the degree distribution is linearly related to the maximum number of balls in any bin $h_m$.
That is, $d_m = 1 + h_m$.
When $k=1$, each ball is assigned to a bin uniformly.
The quantity $h_m$ has been studied extensively in this case, and the result for $k=1$ in Proposition \ref{prop:twochoices} is well-known \cite{johnson1977urn}.
When $k>1$, Algorithm \ref{algo:dreg_approx} exploits the `power of two choices' paradigm.
Power of two choices states that by picking the minimum-degree node among two choices, the maximum in-degree $h_m=\frac{\log \log n}{\log 2}(1+o(1)) + \Theta(1)$ with high probability \cite{azar}. 
This is an exponential reduction in maximum degree compared to when $k=1$.
More generally, for arbitrary $k>1$, the maximum degree is $\frac{\log \log n}{\log k}(1+o(1)) + \Theta(1)$.
This result is due to Azar et al. \cite{azar} and is well-studied in subsequent literature \cite{mitzenmacher2001power}.